\documentclass{IEEEtran}

\IEEEoverridecommandlockouts

\usepackage{cite}
\usepackage{amsmath,amssymb,amsfonts,amsthm,newtxmath,mleftright}
\usepackage{graphicx,xcolor}
\usepackage[caption=false,font=footnotesize]{subfig}
\usepackage{fancyhdr}

\renewcommand{\vec}[1]{\ensuremath{\mathbf{#1}}} 
\newcommand{\mat}[1]{\ensuremath{\mathbf{#1}}} 
\newcommand{\T}{\ensuremath{\mathsf{T}}} 

\newcommand{\expval}[1]{\operatorname{E}[#1]}
\DeclareMathOperator{\atantwo}{atan2}

\newcommand{\pd}{\ensuremath{p_\mathit{d}}}
\newcommand{\pfa}{\ensuremath{p_\mathit{fa}}}

\newtheorem{proposition}{Proposition}

\theoremstyle{remark}
\newtheorem*{remark}{Remark}

\begin{document}
	
\title{Structured Covariance Matrix Estimation for Noise-Type Radars}

\author{
	\IEEEauthorblockN{David Luong,~\IEEEmembership{Graduate Student Member, IEEE}, Bhashyam Balaji,~\IEEEmembership{Senior Member, IEEE}, and \\ Sreeraman Rajan,~\IEEEmembership{Senior Member, IEEE}}
	
	\thanks{D.\ Luong is with Carleton University, Ottawa, ON, Canada K1S 5B6. Email: david.luong3@carleton.ca.}
	\thanks{B.\ Balaji is with Defence Research and Development Canada, Ottawa, ON, Canada K2K 2Y7. Email: bhashyam.balaji@drdc-rddc.gc.ca.}
	\thanks{S.\ Rajan is with Carleton University, Ottawa, ON, Canada K1S 5B6. Email: sreeraman.rajan@carleton.ca.}
}

\maketitle

\begin{abstract}
	Standard noise radars, as well as noise-type radars such as quantum two-mode squeezing radar, are characterized by a covariance matrix with a very specific structure. This matrix has four independent parameters: the amplitude of the received signal, the amplitude of the internal signal used for matched filtering, the correlation between the two signals, and the relative phase between them. In this paper, we derive estimators for these four parameters using two techniques. The first is based on minimizing the Frobenius norm between the structured covariance matrix and the sample covariance matrix; the second is maximum likelihood parameter estimation. The two techniques yield the same estimators. We then give probability density functions (PDFs) for all four estimators. Because some of these PDFs are quite complicated, we also provide approximate PDFs. Finally, we apply our results to the problem of target detection and derive expressions for the receiver operating characteristic curves of two different noise radar detectors.
\end{abstract}


\section{Introduction}

The very name \emph{noise radar} suggests the nature of its transmit signal: noise \cite{cooper1967random,narayanan1998design,lukin1998millimeter,dawood2001roc,narayanan2002UWB,tarchi2010SAR,kulpa2013signal,wasserzier2019noise,savci2020noise}. This sets it apart from other types of radars, such as frequency-modulated continuous-wave (FMCW) radars, whose transmit signals are deterministic. There is no denying that FMCW radars are more popular than noise radars. However, from a practical perspective, the randomness of their transmit signals endows them with desirable properties: low probability of intercept, immunity against noise and jamming, and a ``thumbtack'' ambiguity function \cite{thayaparan2006noise,narayanan2016noise}. For these reasons, there has always been a latent undercurrent of research aimed at building noise radars \cite{narayanan2004design,stove2016design,savci2019trials}. But there is a second reason why noise radars are a worthwhile subject of research. There exists at least one other type of radar whose transmit signal is also nondeterministic: \emph{quantum two-mode squeezing} (QTMS) radar, a type of quantum radar \cite{chang2018quantum,luong2019roc}. It turns out that noise radars are closely allied with QTMS radars \cite{luong2019cov}, which links them to quantum radars more generally. This motivates us to examine the theory of noise radar more carefully.

Until recently, quantum radars were confined to the realm of theory \cite{lloyd2008qi,tan2008quantum,barzanjeh2015mqi,wilde2017gaussian} except for a handful of quantum lidar experiments \cite{lopaeva2013qi,zhuang2017entanglementlidars,england2018quantum}. However, in 2018, a team led by Wilson at the Institute for Quantum Computing (University of Waterloo) demonstrated the viability of a \emph{quantum-enhanced noise radar} at microwave frequencies \cite{chang2018quantum}. This experiment was later analyzed using more conventional radar engineering metrics, and \cite{luong2019roc} was the first scientific publication in the world to publish \emph{receiver operating characteristic} (ROC) \emph{curves} for a quantum radar experiment. This experiment, whose leading results were later confirmed by a similar experiment at the Institute of Science and Technology Austria \cite{barzanjeh2019experimental}, showed that microwave quantum radars can be built in the lab. 

Although we introduced the term \emph{QTMS radar} in \cite{luong2019roc} to emphasize the vastly different technology underlying the new quantum radar design, the term \emph{quantum-enhanced noise radar} highlights the theoretical similarities between QTMS radars and standard noise radars. Where detection performance is concerned, we can speak of them collectively as ``noise-type radars''. The main theoretical result that ties noise-type radars together is that they are characterized by a covariance matrix with a very specific structure \cite{luong2019cov}. The matrix depends on four parameters: the amplitude (or power) of the received signal, the amplitude of the internal signal used for matched filtering, the correlation coefficient between the two signals, and the relative phase between the signals. 

In previous work, we highlighted the importance of the correlation coefficient for target detection, and investigated a method for estimating the correlation coefficient \cite{luong2019rice}. This method was based on minimizing the Frobenius norm between the structured covariance matrix and the sample covariance matrix, the latter being calculated directly from the measurement data. The minimization was performed numerically, which is not practical in many radar systems. In this paper, we show that this minimization can be done analytically, which greatly increases the applicability of our results to real-world systems. We exhibit the exact, closed form estimate not only for the correlation coefficient, but for all four parameters in the noise radar covariance matrix. We also show that, by a curious coincidence, the same estimates are obtained via maximum likelihood parameter estimation.

The remainder of this paper is organized as follows. In Sec.\ \ref{sec:background}, we introduce the covariance matrix that characterizes noise-type radars. In Sec.\ \ref{sec:estimating}, we give estimators for the four parameters in the covariance matrix. (The relevant proofs, however, have been relegated to the Appendixes.) In Sec.\ \ref{sec:pdfs}, we characterize the probability distributions of the estimators. Since some of these distributions are complicated, we also give approximations. In Sec.\ \ref{sec:target_detection}, we use these results to analyze the detection performance of noise-type radars. Sec.\ \ref{sec:conclusion} concludes the paper.

\section{The Covariance Matrix for Noise-Type Radars}
\label{sec:background}

In \cite{luong2019cov}, we showed that, under certain conditions, noise-type radars are completely described by a $4 \times 4$ covariance matrix which we will now describe.

It is well known that an electromagnetic signal can be described by a pair of real-valued time series, namely the \emph{in-phase} and \emph{quadrature} voltages of the signal. A noise-type radar, in the simplest case, has two signals associated with it (for a total of four time series): the signal received by the radar and a signal retained within the radar as a reference for matched filtering. We will denote by $I_1[n]$ and $Q_1[n]$ the in-phase and quadrature voltages, respectively, of the received signal. Similarly, let $I_2[n]$ and $Q_2[n]$ denote the in-phase and quadrature voltages of the reference signal. We assume that these voltages are digitized, so these are discrete time series indexed by $n$.

Note that the \emph{transmitted} signal is not explicitly modeled here. All knowledge of the transmitted signal is encoded in the reference signal. The latter may be thought of as a ``copy'' of the transmitted signal, though it is important to note that this copy is necessarily imperfect. The uncertainty principle of quantum mechanics, as applied to in-phase and quadrature voltages, guarantees the existence of a certain amount of error between the transmitted and reference signals \cite{luong2020magazine}. This minimum error manifests itself as noise, which may be termed \emph{quantum noise}.

We now make the assumption that justifies the name ``noise radar'': we assume that the transmitted and reference signals are stationary Gaussian white noise processes with zero mean. We also make the assumption that any other source of noise, such as system noise or atmospheric noise, may be modeled as additive white Gaussian noise. (Note that quantum noise is known to be Gaussian.) Consequently, the received signal is also a stationary Gaussian white noise process. In short, the four time series $I_1[n]$, $Q_1[n]$, $I_2[n]$, and $Q_2[n]$ are real-valued, zero-mean, stationary Gaussian white noise processes; this allows us to simplify the notation by dropping the index $n$. Finally, we assume that these four processes are pairwise independent unless the time lag between the voltages is zero.

Under the above conditions, the received and reference signals of a QTMS radar are fully specified by the $4 \times 4$ covariance matrix $\expval{\vec{x}\vec{x}^\T}$, where $\vec{x} = [I_1, Q_1, I_2, Q_2]^\T$. In \cite{luong2019cov}, we proved that this matrix has a very specific structure. In block matrix format, we may write it as
\begin{equation} \label{eq:QTMS_cov}
	\mat{\Sigma}(\sigma_1, \sigma_2, \rho, \phi) =
	\begin{bmatrix}
		\sigma_1^2 \mat{1}_2 & \rho \sigma_1 \sigma_2 \mat{R}'(\phi) \\
		\rho \sigma_1 \sigma_2 \mat{R}'(\phi)^\T & \sigma_2^2 \mat{1}_2
	\end{bmatrix}
\end{equation}
where $\sigma_1^2$ and $\sigma_2^2$ are the received and reference signal powers, respectively, while $\rho$ is a correlation coefficient, $\phi$ is the phase shift between the signals, $\mat{1}_2$ is the $2 \times 2$ identity matrix, and $\mat{R}'(\phi)$ is the reflection matrix 
\begin{equation}
	\mat{R}'(\phi) = 
	\begin{bmatrix}
		\cos \phi & \sin \phi \\
		\sin \phi & -\cos \phi
	\end{bmatrix} \! .
\end{equation}
Standard noise radars are described by a matrix of the same overall form, but with the rotation matrix
\begin{equation}
	\mat{R}(\phi) = 
	\begin{bmatrix}
		\cos \phi & \sin \phi \\
		-\sin \phi & \cos \phi
	\end{bmatrix}
\end{equation}
taking the place of the reflection matrix. The results in this paper hold for both standard noise radars and QTMS radars after appropriate choices of sign as detailed below. We assume $\sigma_1 \geq 0$, $\sigma_2 \geq 0$, and $\rho \geq 0$ because their signs can always be accounted for by an appropriate choice of $\phi$.

The contribution of this paper is the derivation of estimators for $\sigma_1$, $\sigma_2$, $\rho$, and $\phi$, as well as the presentation of results related to these estimators.

\section{Estimating the Parameters of the Covariance Matrix}
\label{sec:estimating}

We will estimate the four parameters in \eqref{eq:QTMS_cov} via two methods. The first is a ``naive'' method which we might term the \emph{minimum Frobenius norm} (MFN) method. The second is maximum likelihood (ML) estimation.

Both methods start with the sample covariance matrix
\begin{align} \label{eq:sample_cov}
	\hat{\mat{S}} = \frac{1}{N} \sum_{n=1}^N \vec{x}[n] \vec{x}[n]^\T\!,
\end{align}
calculated from $N$ instances of the random vector $\vec{x}$---that is, $N$ samples each from the in-phase and quadrature voltages of the received and reference signals. In radar terminology, we say that we integrate over $N$ samples of the radar's measurement data. Note that, as a consequence of the assumptions outlined in Sec.\ \ref{sec:background}, each sample is independent and identically distributed.

In the following, we will use an overline to denote the sample mean over $N$ samples. For example, $\hat{\mat{S}} = \overline{\vec{x}\vec{x}^\T}$.

\subsection{Minimum Frobenius Norm Estimation}
\label{subsec:MFN_est}

The MFN method consists of minimizing the Frobenius norm between the structured covariance matrix \eqref{eq:QTMS_cov} and the sample covariance matrix \eqref{eq:sample_cov}. More concretely, we perform the minimization
\begin{equation} \label{eq:minimization}
	\min_{\sigma_1, \sigma_2, \rho, \phi} \mleft\| \mat{\Sigma}(\sigma_1, \sigma_2, \rho, \phi) - \hat{\mat{S}} \mright\|_F
\end{equation}
subject to the constraints $0 \leq \sigma_1$, $0 \leq \sigma_2$, and $0 \leq \rho \leq 1$. (The subscript $F$ denotes the Frobenius norm.) The MFN estimators $\hat{\sigma}_1$, $\hat{\sigma}_2$, $\hat{\rho}$, and $\hat{\phi}$ are the arguments which minimize \eqref{eq:minimization}.

In \cite{luong2019rice}, we obtained estimates of $\rho$ by performing the minimization \eqref{eq:minimization} numerically. This procedure is computationally expensive and would be impractical in many radar setups. The results in this paper allow us to do away with numerical optimization altogether.

\subsection{Maximum Likelihood Estimation}

The probability density function for a 4-dimensional multivariate normal distribution with zero mean and covariance matrix $\mat{\Sigma}$ is
\begin{equation}
	f(\vec{x}|\mat{\Sigma}) = \frac{\exp \mleft( -\frac{1}{2} \vec{x}^\T \mat{\Sigma}^{-1} \vec{x} \mright)}{\sqrt{(2 \pi)^4 |\mat{\Sigma}|}}
\end{equation}
where $|\mat{\Sigma}|$ is the determinant of $\mat{\Sigma}$. When considered as a function of $\mat{\Sigma}$ instead of $\vec{x}$, this becomes the likelihood function. The ML estimators arise from maximizing the likelihood function, or equivalently, the log-likelihood function. For $N$ independently drawn samples $\vec{x}[1], \dots, \vec{x}[N]$, the log-likelihood is
\begin{equation} \label{eq:log_l}
	\ell(\mat{\Sigma}) = -\frac{N}{2} \mleft( \ln |\mat{\Sigma}| + 4 \ln(2 \pi) - \overline{\vec{x}^\T \mat{\Sigma}^{-1} \vec{x}} \mright).
\end{equation}

\subsection{Parameter Estimates}

One of the main results of this paper, and perhaps the most surprising of them, is that the MFN and ML methods lead to the same estimators. We will relegate the actual derivations of the estimators to the Appendixes. Here we present only the final result, namely the estimators themselves as obtained from both methods.

In order to express the estimators in a compact form, we introduce the following auxiliary quantities:
\begin{subequations}
\begin{align}
	\label{eq:aux_P1}
	P_1 &= I_1^2 + Q_1^2 \\
	\label{eq:aux_P2}
	P_2 &= I_2^2 + Q_2^2 \\
	\label{eq:aux_Rc}
	R_c &= I_1 I_2 \mp Q_1 Q_2 \\
	\label{eq:aux_Rs}
	R_s &= I_1 Q_2 \pm I_2 Q_1.
\end{align}
\end{subequations}
For $R_c$ and $R_s$, the upper signs apply when the reflection matrix $\mat{R}'(\phi)$ is used in \eqref{eq:QTMS_cov} (QTMS radar); the lower signs apply when the rotation matrix $\mat{R}(\phi)$ is used (standard noise radar). Note that $\bar{P}_1$, $\bar{P}_2$, $\bar{R}_c$, and $\bar{R}_s$ are merely sums of the appropriate entries in the sample covariance matrix $\hat{\mat{S}}$.

\begin{proposition}
	In terms of the auxiliary quantities \eqref{eq:aux_P1}--\eqref{eq:aux_Rs}, the MFN and ML estimators for the four parameters in \eqref{eq:QTMS_cov} are
	\begin{subequations}
		\begin{align}
			\label{eq:est_sigma1}
			\hat{\sigma}_1 &= \sqrt{ \frac{\bar{P}_1}{2} } \\
			\label{eq:est_sigma2}
			\hat{\sigma}_2 &= \sqrt{ \frac{\bar{P}_2}{2} } \\
			\label{eq:est_rho}
			\hat{\rho} &= \sqrt{ \frac{\bar{R}_c^2 + \bar{R}_s^2}{\bar{P}_1 \bar{P}_2} } \\
			\label{eq:est_phi}
			\hat{\phi} &= \atantwo(\bar{R}_s, \bar{R}_c)
		\end{align}
	\end{subequations}
	where $\atantwo(y, x)$ is the two-argument arctangent.
\end{proposition}

\begin{proof}
	See Appendix \ref{app:mfn} for a proof that these are the MFN estimators, and Appendix \ref{app:ml} for a proof that these same estimators are also the ML estimators.
\end{proof}

\section{Probability Distributions for the Parameter Estimates}
\label{sec:pdfs}

In this section, we give expressions for the probability density functions (PDFs) of the estimators \eqref{eq:est_sigma1}--\eqref{eq:est_phi}. Of these, the most important is perhaps the one for $\hat{\rho}$ because of its importance for target detection, a connection which we will explore in Sec.\ \ref{sec:target_detection}. However, for completeness, we give PDFs for all four estimators. 

For $\hat{\rho}$ and $\hat{\phi}$, the exact PDFs are quite complicated, so we will give simple approximations to these distributions. In order to quantify the goodness of these approximations, we will make use of a metric on probability distributions known as the \emph{total variation distance} (TVD). Informally speaking, the TVD between two probability distributions is defined as the maximum possible difference between the probabilities assigned to the same event by the two distributions. It always lies in the interval $[0,1]$. According to Lemma 2.1 of \cite{tsybakov2009introduction}, when the distributions are described by PDFs, the TVD is
\begin{equation} \label{eq:totVarDist}
	\mathit{TVD} = \frac{1}{2} \int \mleft| f(x) - g(x) \mright| \, dx
\end{equation}
where $f(x)$ and $g(x)$ are the PDFs of the two distributions, and the integral is taken over the whole domain of the PDFs. Apart from furnishing us with a concrete formula for the TVD, this expression gives us a simpler interpretation of the TVD: it is half the integrated absolute error between the PDFs.

\subsection{PDFs for $\hat{\sigma}_1$ and $\hat{\sigma}_2$}

The distributions of the estimated signal amplitudes $\hat{\sigma}_1$ and $\hat{\sigma}_2$ are nothing more than rescaled versions of the chi distribution, as shown in the following proposition.

\begin{proposition}
	The PDF of $\hat{\sigma}_1$ for $x \geq 0$ is
	\begin{equation} \label{eq:PDF_sigma1}
		f_{\hat{\sigma}_1}(x | \sigma_1, N) = \frac{2N^N}{\Gamma(N) \sigma_1^{2N}} x^{2N-1} \exp \mleft( -\frac{N x^2}{\sigma_1^2} \mright)
	\end{equation}
	where $\Gamma(N)$ denotes the gamma function. This also holds for $\hat{\sigma}_2$ when $\sigma_1$ is replaced with $\sigma_2$.
\end{proposition}

\begin{proof}
	Note that $\bar{P}_1$ consist of a sum of squares of $2N$ independent and identically distributed normal random variables, namely $N$ instances each of $I_1$ and $Q_1$. Both $I_1$ and $Q_1$ have zero mean and standard deviation $\sigma_1$, as can be seen from \eqref{eq:QTMS_cov}. Thus, the rescaled random variable 
	\begin{equation}
		\sqrt{\frac{2N}{\sigma_1^2}} \hat{\sigma}_1 = \sqrt{\sum_{n=1}^N \mleft( \frac{i_1[n]}{\sigma_1} \mright)^{\!2} + \mleft( \frac{q_1[n]}{\sigma_1} \mright)^{\!2}},
	\end{equation}
	being the positive square root of the sum of squares of $2N$ standard normal variates, follows a chi distribution with $2N$ degrees of freedom. The proposition follows upon applying the standard change of variable formula to the PDF of the chi distribution.
\end{proof}

\begin{remark}
	The PDF \eqref{eq:PDF_sigma1} may be recognized as a Nakagami $m$-distribution \cite{nakagami1960m} with parameters $m = N$ and $\Omega = \sigma_1^2$.
\end{remark}

\begin{figure}[t]
	\centerline{\includegraphics[width=\columnwidth]{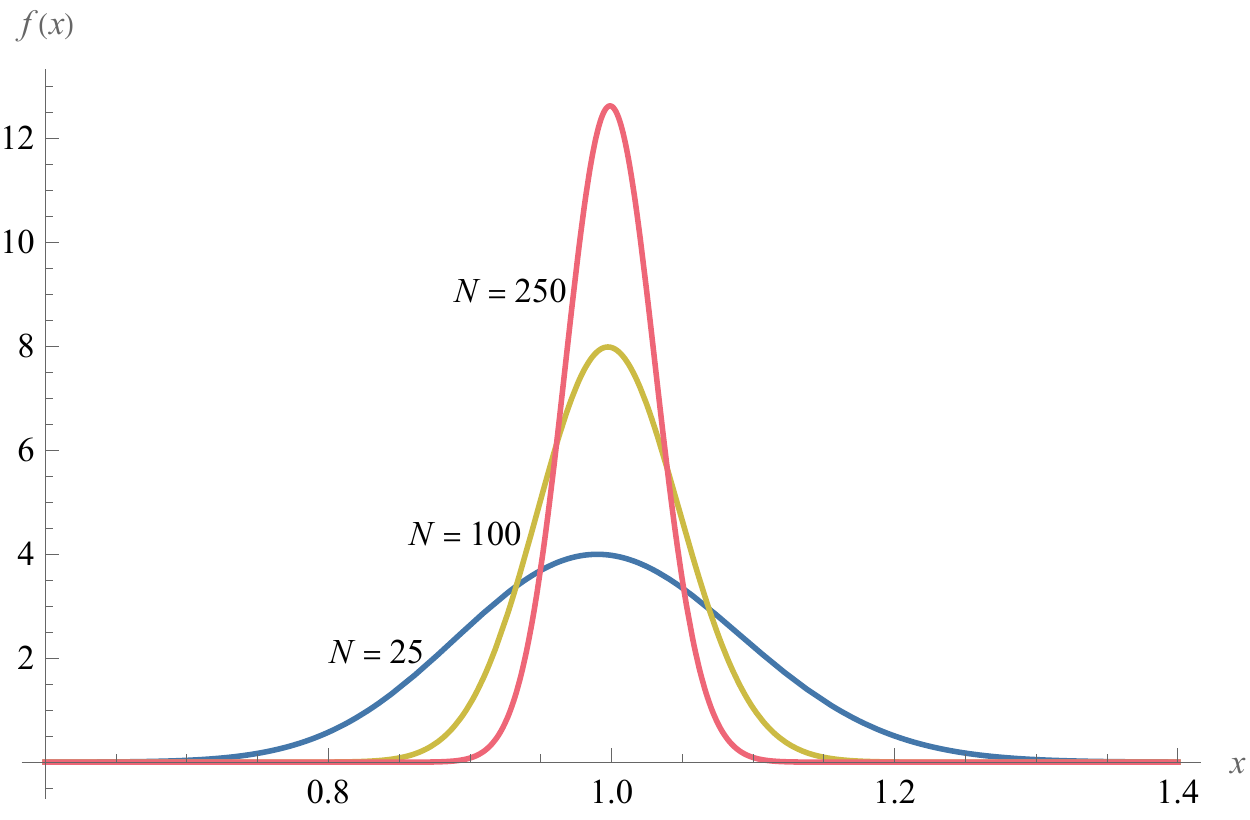}}
	\caption{Probability density function of $\hat{\sigma}_1$ when $\sigma_1 = 1$, $N \in \{25, 100, 250\}$.}
	\label{fig:PDF_sigma1}
\end{figure}

Plots of $f_{\hat{\sigma}_1}(x | \sigma_1, N)$ are shown in Fig.\ \ref{fig:PDF_sigma1}.

\subsection{Exact and Approximate PDFs for $\hat{\rho}$}

The derivation of the PDF for the estimated correlation coefficient $\hat{\rho}$ is extremely involved. But luckily, our task has been done for us. We exploit an intriguing connection between noise radar and the theory of two-channel synthetic aperture radar (SAR), in which matrices analogous to \eqref{eq:QTMS_cov} appear. (Note, however, that the matrices in two-channel SAR are $2 \times 2$ complex-valued matrices instead of $4 \times 4$ real-valued matrices.) In two-channel SAR, the quantity analogous to $\rho$ is known as the \emph{coherence}. An estimator for the coherence, essentially identical to \eqref{eq:est_rho}, was investigated in \cite{touzi1996statistics,touzi1999coherence,gierull2001unbiased,sikaneta2004detection}. We now quote one of their results here.

\begin{proposition}
	When $N > 2$ and $\rho \neq 1$, the PDF of $\hat{\rho}$ for $0 \leq x \leq 1$ is
	\vspace{-\jot}
	\begin{multline} \label{eq:PDF_rho}
		f_{\hat{\rho}}(x | \rho, N) = 2 (N-1)(1-\rho^2)^N \\
			\times x(1-x^2)^{N-2} {}_2F_1(N, N; 1; \rho^2 x^2)
	\end{multline}
	where ${}_2F_1$ is the Gaussian hypergeometric function.
\end{proposition}

\begin{proof}
	See Sec.\ VI of \cite{touzi1996statistics}.
\end{proof}

This expression is both numerically and analytically unwieldy (except when $\rho = 0$). However, we are able to supply an empirical PDF which approximates \eqref{eq:PDF_rho} well when $N$ is larger than approximately 100. In \cite{luong2019rice}, we showed that the correlation coefficients estimated using the MFN method (albeit with a numerical minimization instead of an analytic one) approximately follow a Rice distribution. Recall that the PDF of the Rice distribution is
\begin{equation} \label{eq:PDF_rice}
	f_\text{Rice}(x | \alpha, \beta) = \frac{x}{\beta^2} \exp \mleft( -\frac{x^2 + \alpha^2} {2\beta^2} \mright) I_0 \mleft( \frac{x\alpha}{\beta^2} \mright)
\end{equation}
where $\alpha$ and $\beta$ are the parameters of the distribution, and $I_0$ is the modified Bessel function of the first kind of order zero (not to be confused with the in-phase voltages $I_1$ or $I_2$). The approximation derived in \cite{luong2019rice} may be summarized as follows.

\begin{proposition} \label{prop:approx_rice}
	When $N \gtrapprox 100$, $\hat{\rho}$ approximately follows a Rice distribution with parameters
	\begin{subequations}
		\begin{align}
			\label{eq:rho_rice_alpha}
			\alpha &= \rho \\
			\label{eq:rho_rice_beta}
			\beta &= \frac{1-\rho^2}{\sqrt{2N}}.
		\end{align}
	\end{subequations}
\end{proposition}

\begin{figure}[t]
	\centerline{\includegraphics[width=\columnwidth]{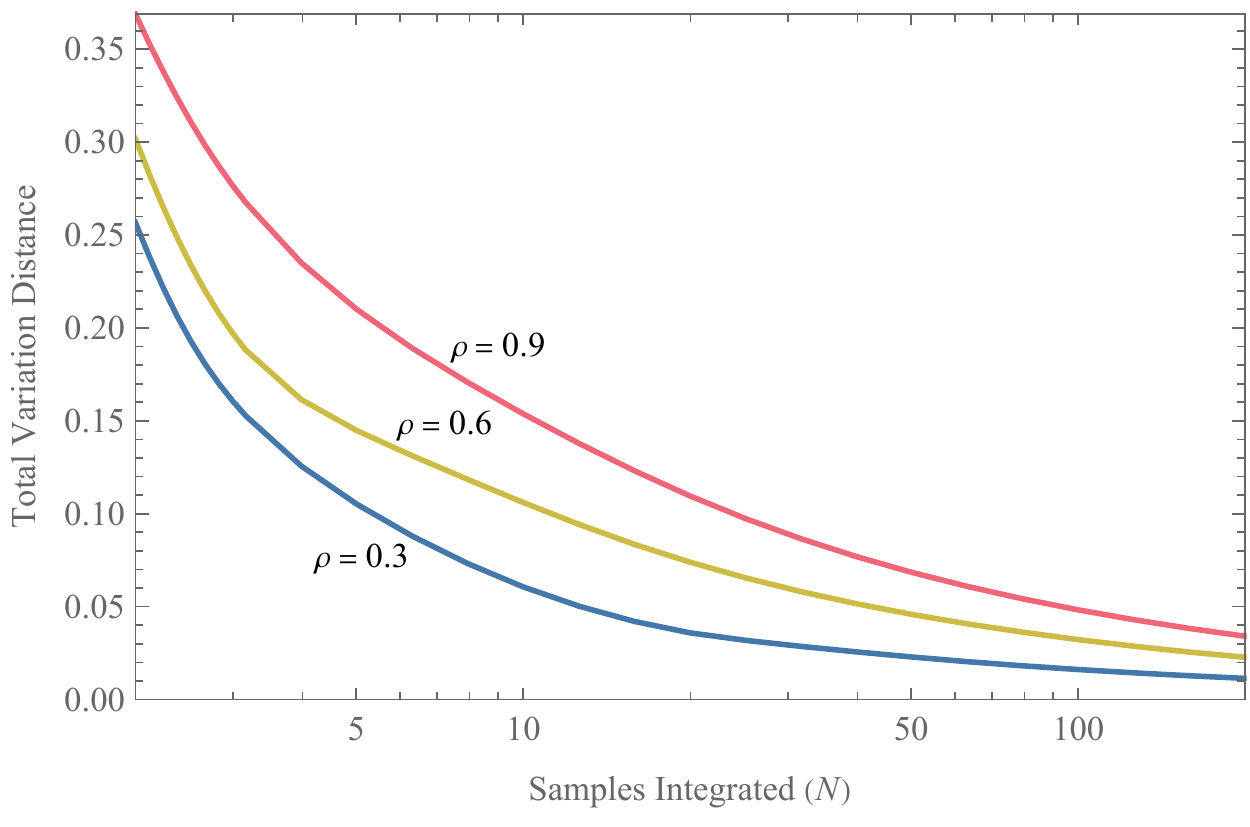}}
	\caption{Total variation distance between the exact probability density function of $\hat{\rho}$ and the approximation described in Proposition \ref{prop:approx_rice}, plotted as a function of $N$, for $\rho \in \{0.3, 0.6, 0.9\}$.}
	\label{fig:totVarDist_rho}
\end{figure}

Because this is an empirical approximation, we can only give plausibility arguments based on numerical results. In Sec.\ V of \cite{luong2019rice}, we showed that this approximation is a good one by simulating radar detection data for various values of $\rho$ and $N$ and fitting Rice PDFs to the resulting histograms. We now build on that work by calculating the total variation distance $\mathit{TVD}_{\hat{\rho}}$ between the exact PDF \eqref{eq:PDF_rho} and the Rician approximation. Fig.\ \ref{fig:totVarDist_rho} shows plots of $\mathit{TVD}_{\hat{\rho}}$ as a function of $N$ for various values of $\rho$. We see that $\mathit{TVD}_{\hat{\rho}}$ increases with $\rho$ and decreases with $N$. At $N = 100$, $\mathit{TVD}_{\hat{\rho}}$ is lower than 0.05 even for $\rho$ as high as 0.9. This is strong evidence that the Rician approximation is indeed a good one when $N \gtrapprox 100$.

\begin{remark}
	Although the expressions \eqref{eq:rho_rice_alpha} and \eqref{eq:rho_rice_beta} were empirically determined, with no basis other than simulations, the fact that $\hat{\rho}$ is approximately Rician for large $N$ has some theoretical grounding. The basic idea is that a Rice distribution is the distribution of the norm of a bivariate normal random vector whose covariance matrix is proportional to the identity. To connect this idea to $\hat{\rho}$, begin by invoking the central limit theorem to approximate $\bar{R}_c$ and $\bar{R}_s$ in \eqref{eq:est_rho} as normally distributed random variables. Next, replace $\bar{P}_1$ and $\bar{P}_2$ with the expected values $\expval{P_1} = 2 \sigma_1^2$ and $\expval{P_2} = 2 \sigma_2^2$, respectively. The result, up to first order in $\rho$, is a Rice-distributed random variable with $\alpha = \rho$ and $\beta = 1/\sqrt{2N}$. For a more detailed development of this argument, see Proposition \ref{prop:approx_detDN} and its proof.
\end{remark}

\begin{figure*}[t]
	\centering
	\subfloat[]{\includegraphics[width=\columnwidth]{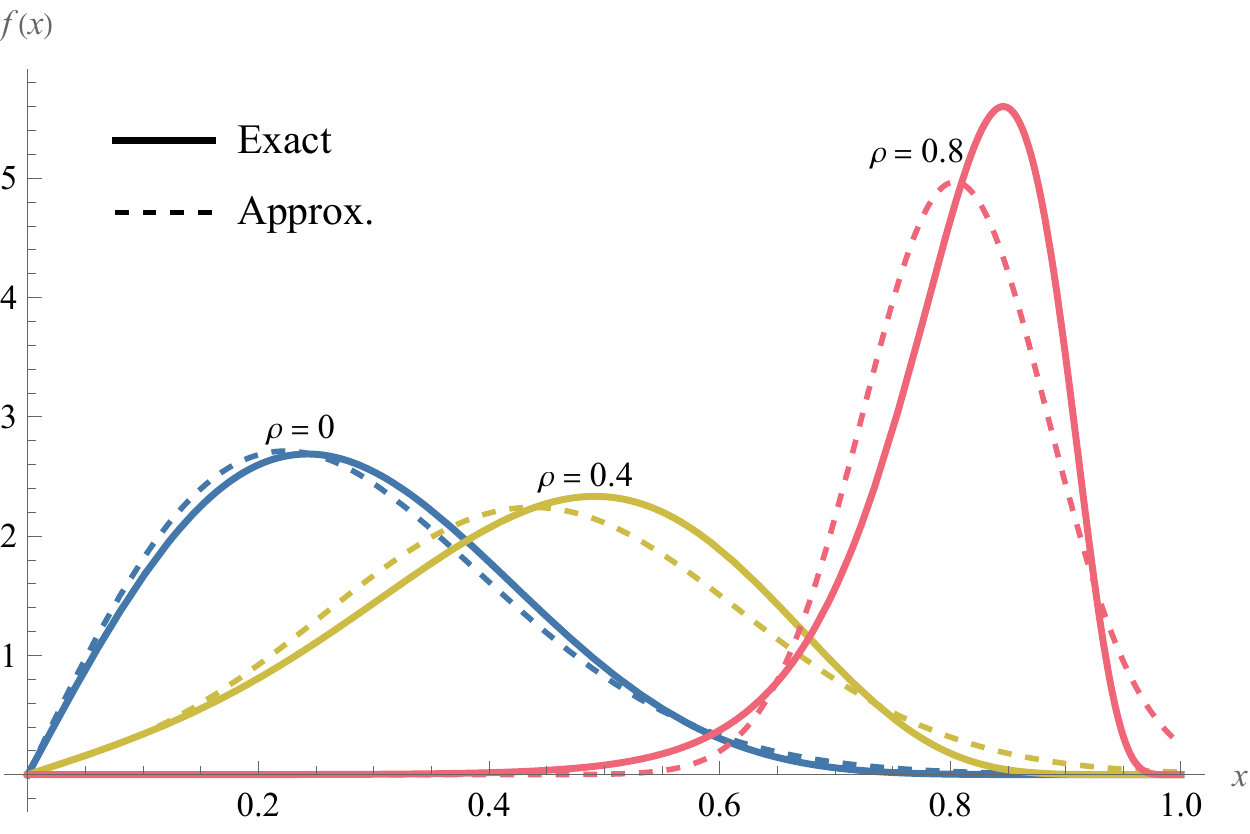}
		\label{subfig:PDF_rho_rho}}
	\hfil
	\subfloat[]{\includegraphics[width=\columnwidth]{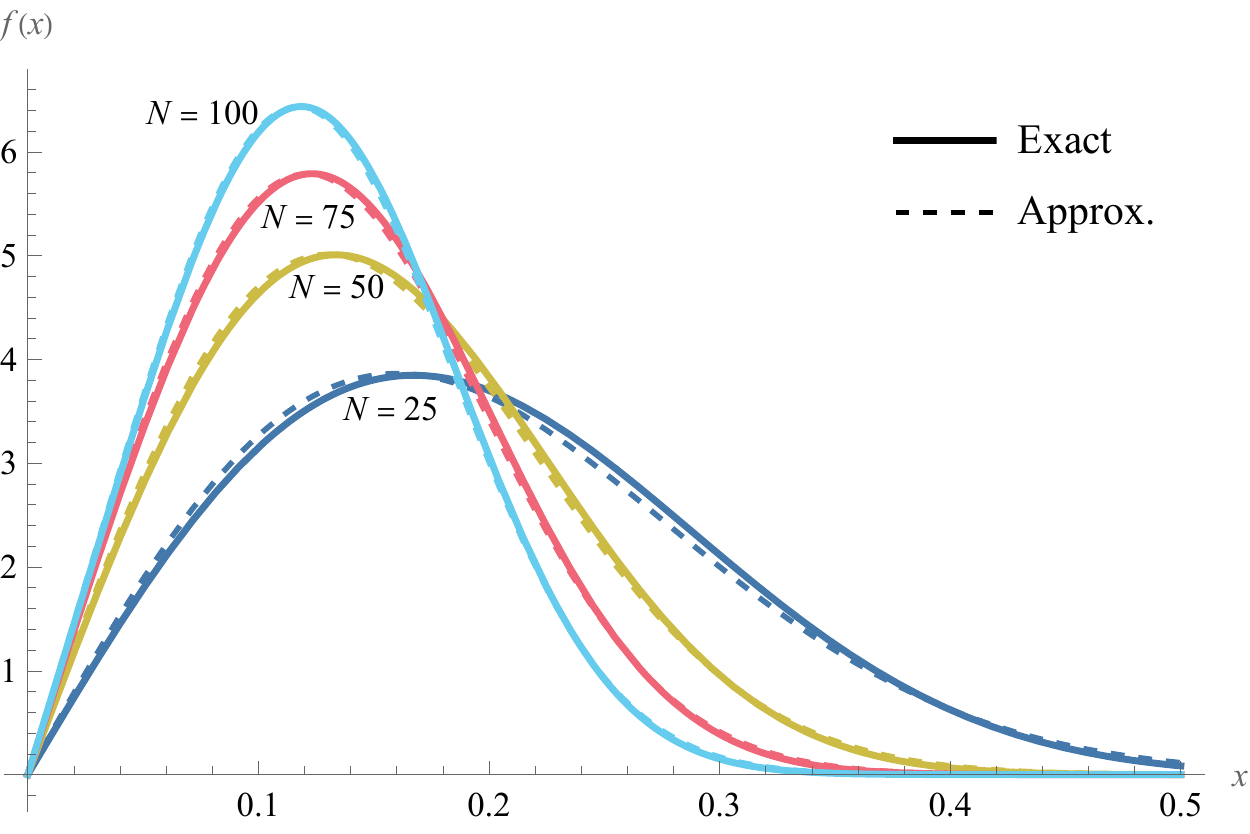}
		\label{subfig:PDF_rho_N}}
	\caption{Probability density function of $\hat{\rho}$, together with the Rice distribution approximations described in Proposition \ref{prop:approx_rice}. In (a), $N = 10$ and $\rho \in \{0, 0.4, 0.8\}$; in (b), $\rho = 0.1$ and $N \in \{25, 50, 75, 100\}$.}
	\label{fig:PDF_rho}
\end{figure*}

In Fig.\ \ref{fig:PDF_rho}, we present plots of $f_{\hat{\rho}}(x | \rho, N)$ for various values of $\rho$ and $N$, together with the Rice distribution approximations. In Fig.\ \ref{subfig:PDF_rho_rho}, we see that the Rice distribution is not always a good fit because $N$ is small. Fig.\ \ref{subfig:PDF_rho_N} shows that the fit becomes quite good as $N$ increases; indeed, at $N = 100$ there is hardly any visible difference between the exact and approximate PDFs.

A word of warning is appropriate here. The Rice distribution approximation outlined in Proposition \ref{prop:approx_rice} must not be confused with the Rice distribution that appears in the context of continuous-wave (CW) radars. It is true that, when a radar transmits a sinusoidal signal and detects using a square-law detector, the detector output is Rice distributed; see e.g.\ Ch.\ 4 of \cite{mahafza2000radar}. However, this is a completely different case from Proposition \ref{prop:approx_rice}. Not only is the transmit signal totally different (sinusoidal waveform vs.\ Gaussian noise), Proposition \ref{prop:approx_rice} describes an \emph{approximation}, whereas the Rice distribution for CW radars is \emph{exact}. In the experience of the authors, the coincidental appearance of the Rice distribution in these two different contexts has led to confusion. Therefore, we emphasize that these two applications of the Rice distribution are unrelated.

\subsection{Exact and Approximate PDFs for $\hat{\phi}$}

Finally, we give the PDF of the estimated phase $\hat{\phi}$. Again, we are able to take over a result from two-channel SAR.

\begin{proposition}
	The PDF of $\hat{\phi}$ is
	\vspace{-\jot}
	\begin{multline} \label{eq:PDF_phi}
		f_{\hat{\phi}}(\theta | \rho, \phi, N) = \frac{\Gamma \big( N + \frac{1}{2} \big) (1 - \rho^2)^N \xi}{2 \sqrt{\pi} \Gamma(N) (1 - \xi^2)^{N + \frac{1}{2}}} 	\\ + \frac{(1-\rho^2)^N}{2\pi} {}_2F_1 \mleft( N, 1; \tfrac{1}{2}; \xi^2 \mright)
	\end{multline}
	where
	\begin{equation}
		\xi \equiv \rho \cos( \theta - \phi ).
	\end{equation}
\end{proposition}

\begin{proof}
	See Sec.\ 2 of \cite{lee1994statistics}. Alternative forms of the PDF are given in \cite[eq.\ (12)]{lopes1992phase} and \cite[eq.\ (10)]{joughin1994probability}.
\end{proof}

\begin{figure}[t]
	\centerline{\includegraphics[width=\columnwidth]{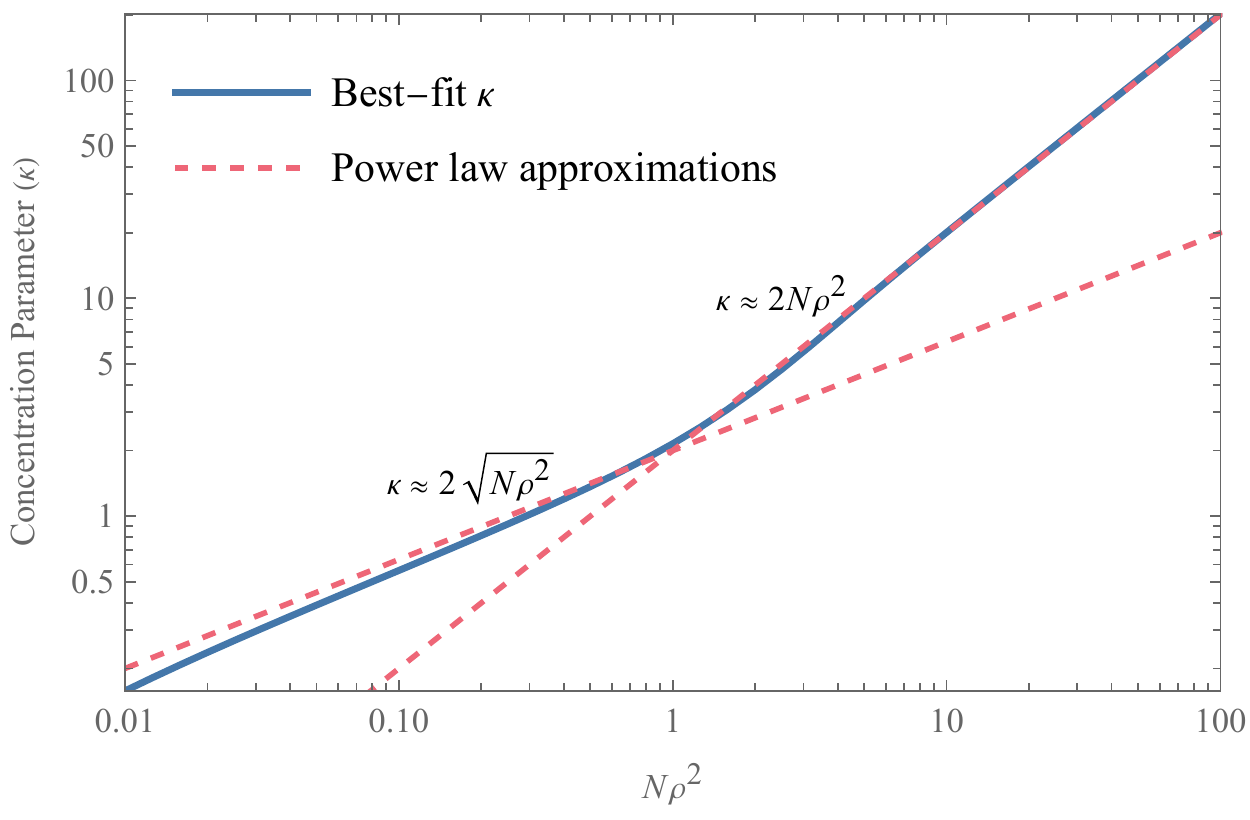}}
	\caption{Concentration parameter $\kappa$ from the von Mises distribution when fitted to the distribution of $\hat{\phi}$, plotted as a function of $N\rho^2$. Also plotted are approximations to the best-fit $\kappa$.}
	\label{fig:kappa_approx}
\end{figure}

This expression is, if anything, even more unwieldy than \eqref{eq:PDF_rho}. However, after plotting the PDF \eqref{eq:PDF_phi} for many values of $\rho$ and $N$, we observed that it always has the same basic shape as the von Mises distribution. This is one of the most basic probability distributions in circular statistics, and can be thought of as the circular analog of the normal distribution. Its PDF is
\begin{equation}
	f(\theta | \mu, \kappa) = \frac{e^{\kappa\cos(\theta - \mu)}}{2\pi I_0(\kappa)},
\end{equation}
where $\mu$ and $\kappa$ are the parameters of the distribution. They correspond to the parameters of the normal distribution in the following sense: when $\kappa \to \infty$, the von Mises distribution approaches the normal distribution with mean $\mu$ and variance $1/\kappa$ (on an appropriate interval of length $2\pi$). Thus, $\mu$ is the mean and $\kappa$ is a ``concentration parameter'': the higher the $\kappa$, the narrower the distribution.

In fitting the von Mises distribution to \eqref{eq:PDF_phi}, choosing $\mu$ is simple enough: since \eqref{eq:PDF_phi} is symmetric about $\phi$, we simply choose $\mu = \phi$. The concentration parameter $\kappa$, however, is less straightforward to choose. To fit a value for $\kappa$, we begin by calculating the so-called ``mean resultant length'',
\begin{equation} \label{eq:phi_R}
	R = \mleft| \int_{-\pi}^{\pi} f_{\hat{\phi}}(\theta | \rho, \phi, N) e^{j\theta} \, d\theta \mright|.
\end{equation}
In \cite{sra2011vonMises}, an approximation of the parameter $\kappa$ is given in terms of the mean resultant length by
\begin{equation} \label{eq:kappa_approx}
	\kappa \approx \frac{R(2 - R^2)}{1 - R^2}.
\end{equation}
In Fig.\ \ref{fig:kappa_approx}, we use \eqref{eq:phi_R} and \eqref{eq:kappa_approx} to plot $\kappa$ as a function of $N\rho^2$. The reason why we plot $\kappa$ against $N\rho^2$ is that $\kappa$ appears to depend on $\rho$ and $N$ only through this combination. This is not evident from \eqref{eq:PDF_phi}, but nevertheless this behavior holds good for a wide variety of values for $\rho$ and $N$. From this plot, we find that when $N\rho^2 \leq 1$, $\kappa \approx 2 \sqrt{N\rho^2}$, otherwise $\kappa \approx 2 N\rho^2$. These approximations are also shown in Fig.\ \ref{fig:kappa_approx}. This leads to the following proposition.

\begin{proposition} \label{prop:approx_vonMises}
	The estimator $\hat{\phi}$ approximately follows a von Mises distribution with parameters
	\begin{subequations}
		\begin{align}
			\mu &= \phi \\
			\kappa &= 
				\begin{cases}
					2 \sqrt{N\rho^2}, & N\rho^2 \leq 1 \\
					2 N\rho^2, & N\rho^2 > 1.
				\end{cases}
		\end{align}
	\end{subequations}
\end{proposition}

\begin{figure}[t]
	\centerline{\includegraphics[width=\columnwidth]{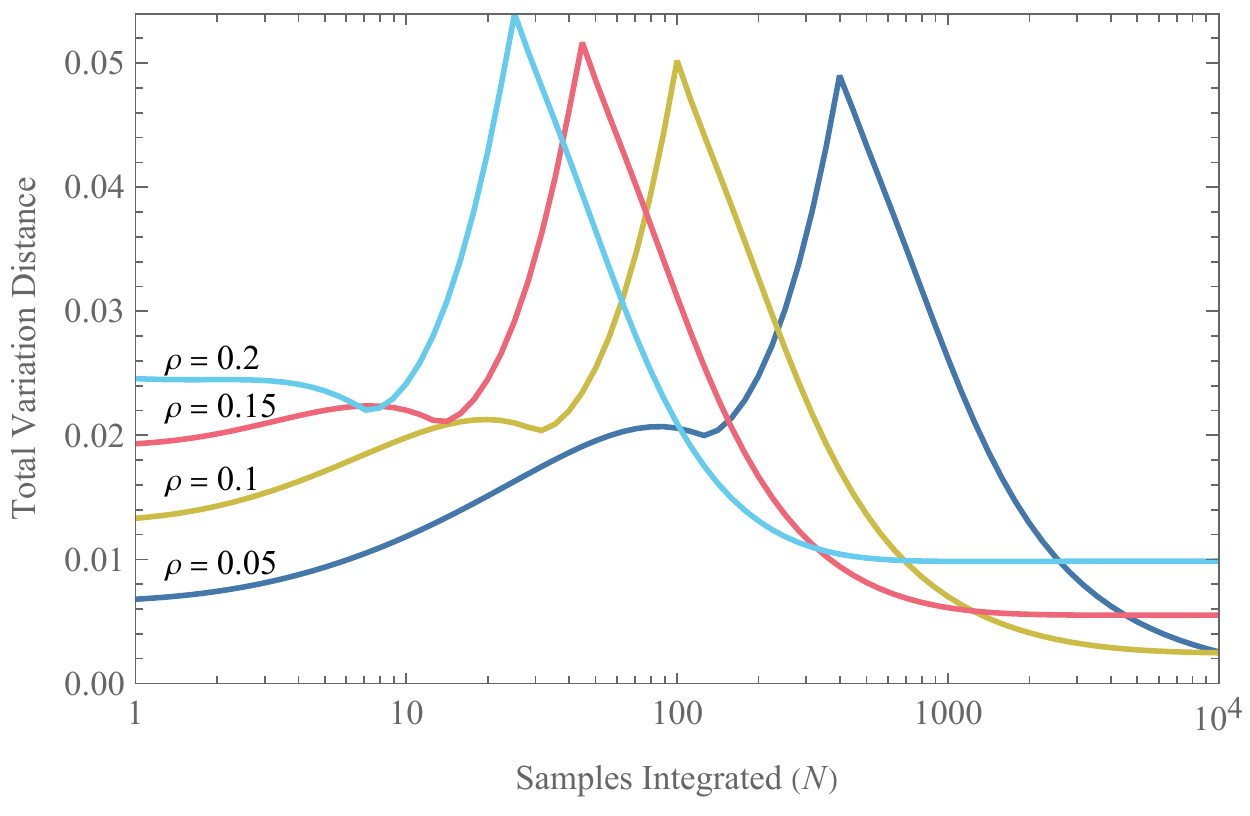}}
	\caption{Total variation distance between the exact probability density function of $\hat{\phi}$ and the approximation described in Proposition \ref{prop:approx_vonMises}, plotted as a function of $N$, for $\rho \in \{0.05, 0.1, 0.15, 0.2\}$.}
	\label{fig:totVarDist_phi}
\end{figure}

\begin{figure}[t]
	\centering
	\subfloat[]{\includegraphics[width=\columnwidth]{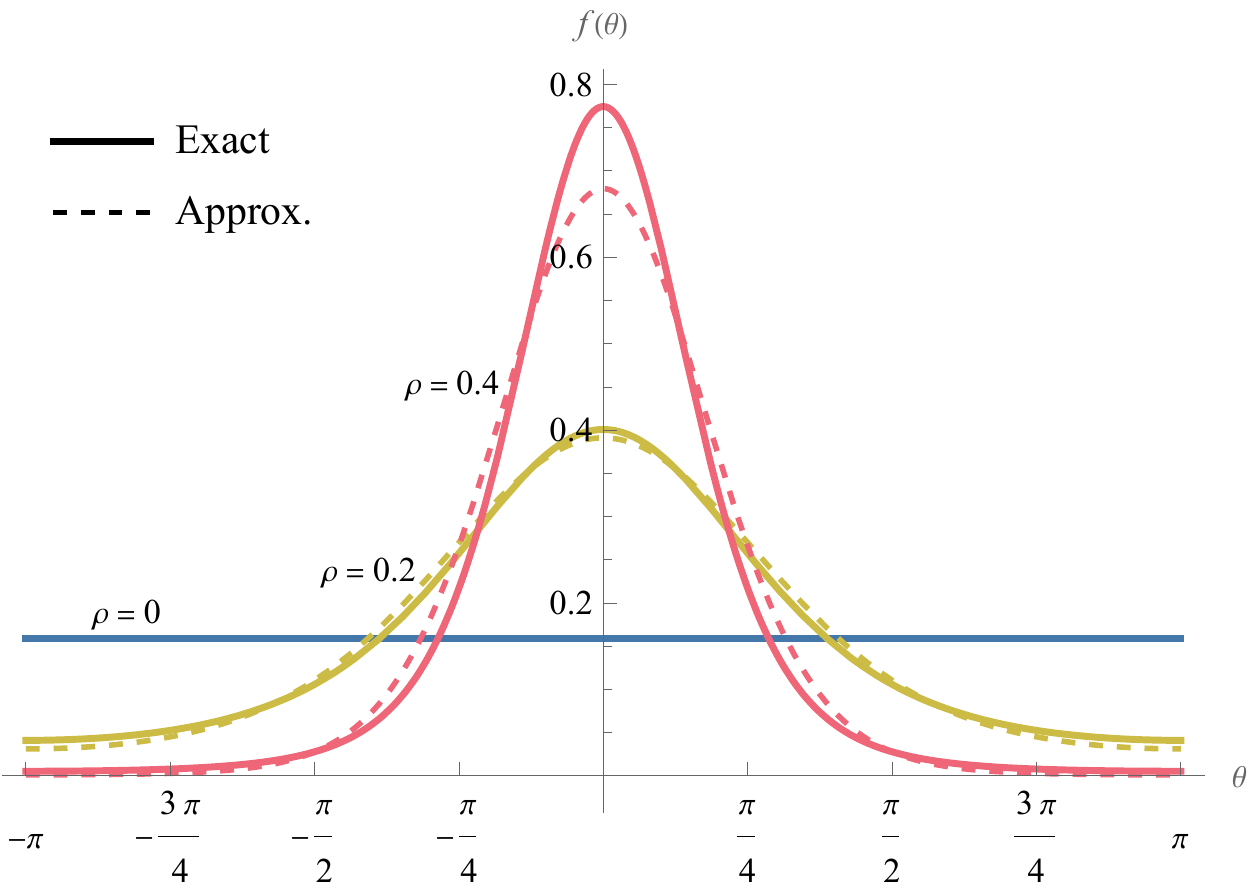}
		\label{subfig:PDF_phi_rho}} \\
	\subfloat[]{\includegraphics[width=\columnwidth]{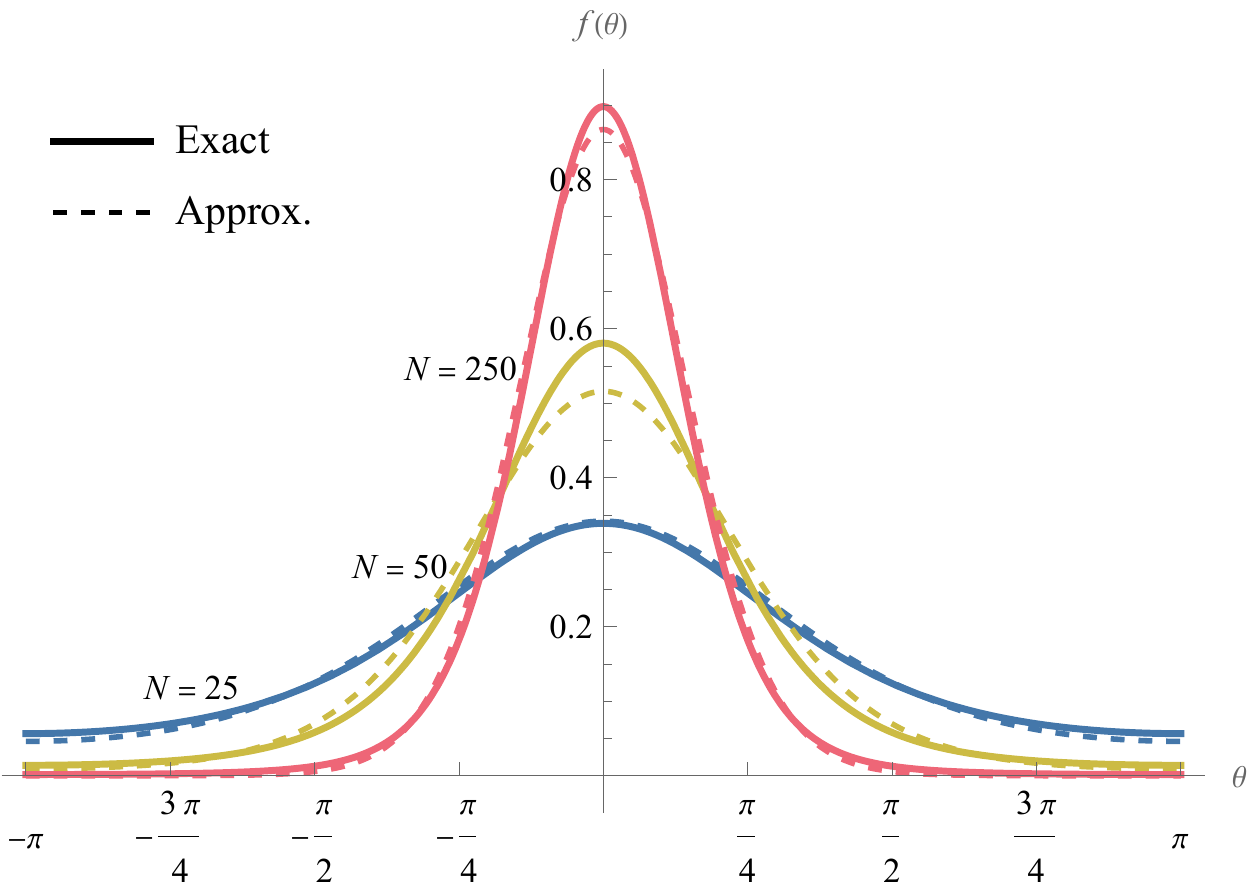}
		\label{subfig:PDF_phi_N}}
	\caption{Probability density function of $\hat{\phi}$, together with the von Mises distribution approximation described in Proposition \ref{prop:approx_vonMises}. In (a), $N = 10$ and $\rho \in \{0, 0.2, 0.4\}$; in (b), $\rho = 0.1$ and $N \in \{25, 50, 250\}$. For all cases, $\phi = 0$.}
	\label{fig:PDF_phi}
\end{figure}

To show the plausibility of this empirical result, we again turn to the TVD. Fig.\ \ref{fig:totVarDist_phi} shows plots of $\mathit{TVD}_{\hat{\phi}}$ as a function of $N$ for various values of $\rho$. (Unfortunately, numerical instabilities prevented us from producing plots when $\rho$ is large, but we expect the behavior to be largely the same.) Unlike $\mathit{TVD}_{\hat{\rho}}$, $\mathit{TVD}_{\hat{\phi}}$ does not appear to decay fully to 0 as $N$ increases. However, an inspection of the vertical axis in Fig.\ \ref{fig:totVarDist_phi} shows that $\mathit{TVD}_{\hat{\phi}}$ is small for \emph{all} values of $N$. There are peaks corresponding to $N\rho^2 = 1$, which may perhaps be expected: this point marks the transition between the square-root and linear regimes in Fig.\ \ref{fig:kappa_approx}. We conclude that Proposition \ref{prop:approx_vonMises} is well-substantiated by numerical evidence.

Fig.\ \ref{fig:PDF_phi} shows plots of $f_{\hat{\phi}}(\theta | \rho, 0, N)$ for various values of $\rho$ and $N$, as well as the corresponding von Mises distribution approximations. (We show only the case $\phi = 0$ because the shape of the plots remains the same for any value of $\phi$; only the location of the peak changes.) In all cases, the exact distribution is well-approximated by a von Mises distribution.

\section{Target Detection and the Correlation Coefficient}
\label{sec:target_detection}

In this section, we apply the preceding results to the analysis of detection performance for noise-type radars. Of the four parameters that appear in \eqref{eq:QTMS_cov}, the correlation coefficient $\rho$ is the most important for target detection. In the absence of clutter, the presence or absence of a target can be reduced to a hypothesis test on $\rho$:
\begin{equation} \label{eq:hypotheses}
	\begin{alignedat}{3}
		H_0&: \rho = 0 &&\quad\text{Target absent} \\
		H_1&: \rho > 0 &&\quad\text{Target present}
	\end{alignedat}
\end{equation}
The reason for this is as follows. If there exists a correlation between the reference and received signals, there must be a target to reflect the transmitted signal to the receiver. If there were no target, the only signal received by the radar would be uncorrelated background noise. Now, it is obvious from the form of \eqref{eq:QTMS_cov} that any correlation between signals can only occur when $\rho > 0$. This explains the form of the hypothesis test \eqref{eq:hypotheses}.

\subsection{Generalized Likelihood Ratio Test}

One of the best-known methods for hypothesis testing is the generalized likelihood ratio (GLR) test. This entails maximizing the likelihood function under the two hypotheses. In previous work, we considered the case where the values of the nuisance parameters $\sigma_1$, $\sigma_2$, and $\phi$ were known \cite{luong2022likelihood}. In this paper, since we have ML estimates for those parameters, we need not make the same assumption. In fact, calculating the GLR test statistic---or the GLR \emph{detector}---is a simple task since we have the ML parameters.

Unlike the complicated GLR detector derived in \cite{luong2022likelihood} under the assumption that $\sigma_1 = \sigma_2 = 1$ and $\phi = 0$, the GLR detector takes on a relatively simple form when all the parameters are unknown. In fact, it is equivalent to $\hat{\rho}$ itself, as we will now prove.

\begin{proposition}
	The GLR test is equivalent to using $\hat{\rho}$ as a test statistic.
\end{proposition}

\begin{proof}
	The GLR test statistic for the hypotheses \eqref{eq:hypotheses} may be written as a difference of log-likelihoods:
	\begin{equation} \label{eq:D_GLR}
		D_\text{GLR} = -2[ \ell(\hat{\sigma}_1, \hat{\sigma}_2, 0, \hat{\phi}) - \ell(\hat{\sigma}_1, \hat{\sigma}_2, \hat{\rho}, \hat{\phi}) ].
	\end{equation}
	Notice that the same estimators appear in both terms. This is permissible because the ML estimates $\hat{\sigma}_1$ and $\hat{\sigma}_2$ are the same under both hypotheses in \eqref{eq:hypotheses}. (The likelihood function does not depend on $\phi$ when $\rho = 0$, so it does not matter what value of $\phi$ is substituted.) See Appendix \ref{app:ml} for details.
	
	Substituting \eqref{eq:est_sigma1}--\eqref{eq:est_phi} into \eqref{eq:log_l}, we obtain
	\begin{align}
		D_\text{GLR} &= 2N \ln \mleft( \frac{\bar{P}_1 \bar{P}_2}{\bar{P}_1 \bar{P}_2 - \bar{R}_c^2 - \bar{R}_s^2} \mright) \nonumber \\
			&= -2N \ln(1 - \hat{\rho}^2).
	\end{align}
	This is a strictly increasing function of $\hat{\rho}$. Since applying a strictly increasing function to a test statistic is equivalent to reparameterizing the decision threshold, the test itself does not change. The proposition follows.
\end{proof}

The gold standard for evaluating radar detection performance is the ROC curve, which plots the probability of detection $\pd$ against the probability of false alarm $\pfa$. In the case where $\hat{\rho}$ is used as a detector, obtaining the exact ROC curve requires an integration of \eqref{eq:PDF_rho}, which is extremely difficult. However, with the help of Proposition \ref{prop:approx_rice}, we can derive a closed-form approximation of the ROC curve.

\begin{proposition}
	When $N \gtrapprox 100$, the ROC curve for the $\hat{\rho}$ detector is
	\begin{equation} \label{eq:ROC_rho}
		\pd(\pfa | \rho, N) = Q_1 \mleft( \frac{\rho \sqrt{2N}}{1 - \rho^2}, \frac{\sqrt{2N \big( 1 - \pfa^{1/(N-1)} \big)}}{1 - \rho^2} \mright).
	\end{equation}
	where $Q_1(\cdot, \cdot)$ is the Marcum $Q$-function of order 1 (not to be confused with the quadrature voltage $Q_1$).
\end{proposition}

\begin{proof}
	In the case where $\rho = 0$, the hypergeometric function in \eqref{eq:PDF_rho} drops out and it is possible to integrate the expression directly, yielding the cumulative density function (CDF)
	\begin{equation}
		F_{\hat{\rho}}(x | 0, N) = 1 - (1 - x^2)^{N-1}.
	\end{equation}
	For a given detection threshold $T$, the probability of false alarm is the probability that $\hat{\rho} > T$ given that $\rho = 0$. This is given by
	\begin{equation} \label{eq:det_rho_pFA}
		\pfa(T) = 1 - F_{\hat{\rho}}(x | 0, N) = (1 - T^2)^{N-1}.
	\end{equation}
	Inverting this, we obtain
	\begin{equation} \label{eq:det_rho_threshold}
		T = \sqrt{1 - \pfa^{1/(N-1)}}.
	\end{equation}
	Because $\hat{\rho} \geq 0$, we retain only the positive square root.
	
	To obtain the probability of detection, we make use of the Rician approximation described in Proposition \ref{prop:approx_rice}. The CDF of the Rice distribution is
	\begin{equation}
		F_\text{Rice}(x | \alpha, \beta) = 1 - Q_1 \mleft( \frac{\alpha}{\beta}, \frac{x}{\beta} \mright).
	\end{equation}
	Substituting \eqref{eq:rho_rice_alpha} and \eqref{eq:rho_rice_beta} yields
	\begin{equation} \label{eq:rho_cdf}
		F(x | \rho, N) = 1 - Q_1 \mleft( \frac{\rho \sqrt{2N}}{1 - \rho^2}, \frac{x \sqrt{2N}}{1 - \rho^2} \mright).
	\end{equation}
	The probability of detection is
	\begin{equation}
		\pd(T) = 1 - F(T | \rho, N);
	\end{equation}
	the proposition follows upon substituting \eqref{eq:det_rho_threshold}.
\end{proof}

\begin{remark}
	In \cite{luong2019rice}, a slightly different expression for the ROC curve was derived:
	\begin{equation} \label{eq:ROC_rho_old}
		\pd(\pfa | \rho, N) = Q_1 \mleft( \frac{\rho \sqrt{2N}}{1 - \rho^2}, \frac{\sqrt{-2 \ln \pfa}}{1 - \rho^2} \mright).
	\end{equation}
	This form arises from using the Rician approximation to calculate both $\pd$ and $\pfa$. In the above proposition, we have replaced the latter with the exact value of $\pfa$. There is, however, not much difference between the two for large $N$. The reader may notice a curious connection between \eqref{eq:det_rho_pFA}, the appearance of $\ln \pfa$ in \eqref{eq:ROC_rho_old}, and the well-known representation of the exponential function as a limit, $e^x = \lim_{N \to \infty} (1 + x/N)^N$.
\end{remark}

\begin{figure*}[t]
	\centering
	\subfloat[]{\includegraphics[width=\columnwidth]{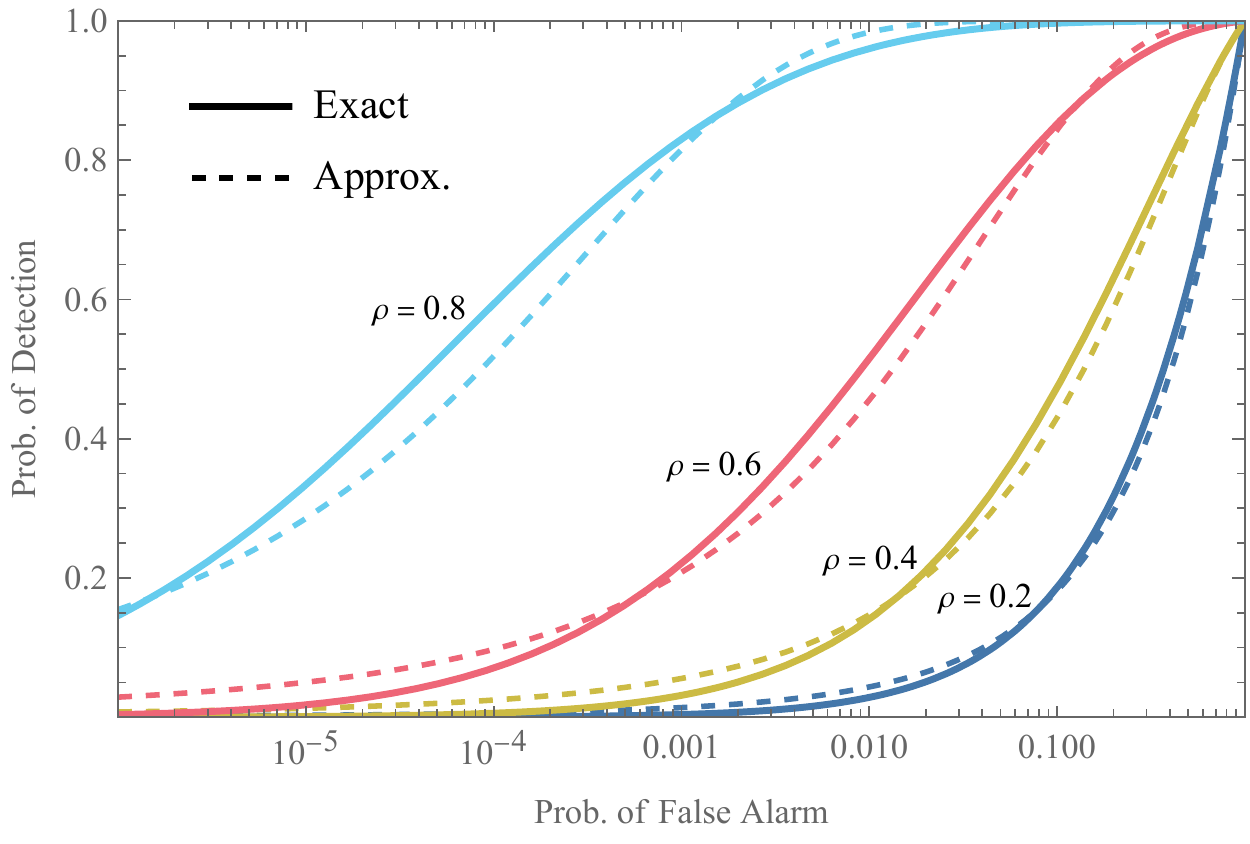}
		\label{subfig:ROC_rho_rho}}
	\hfil
	\subfloat[]{\includegraphics[width=\columnwidth]{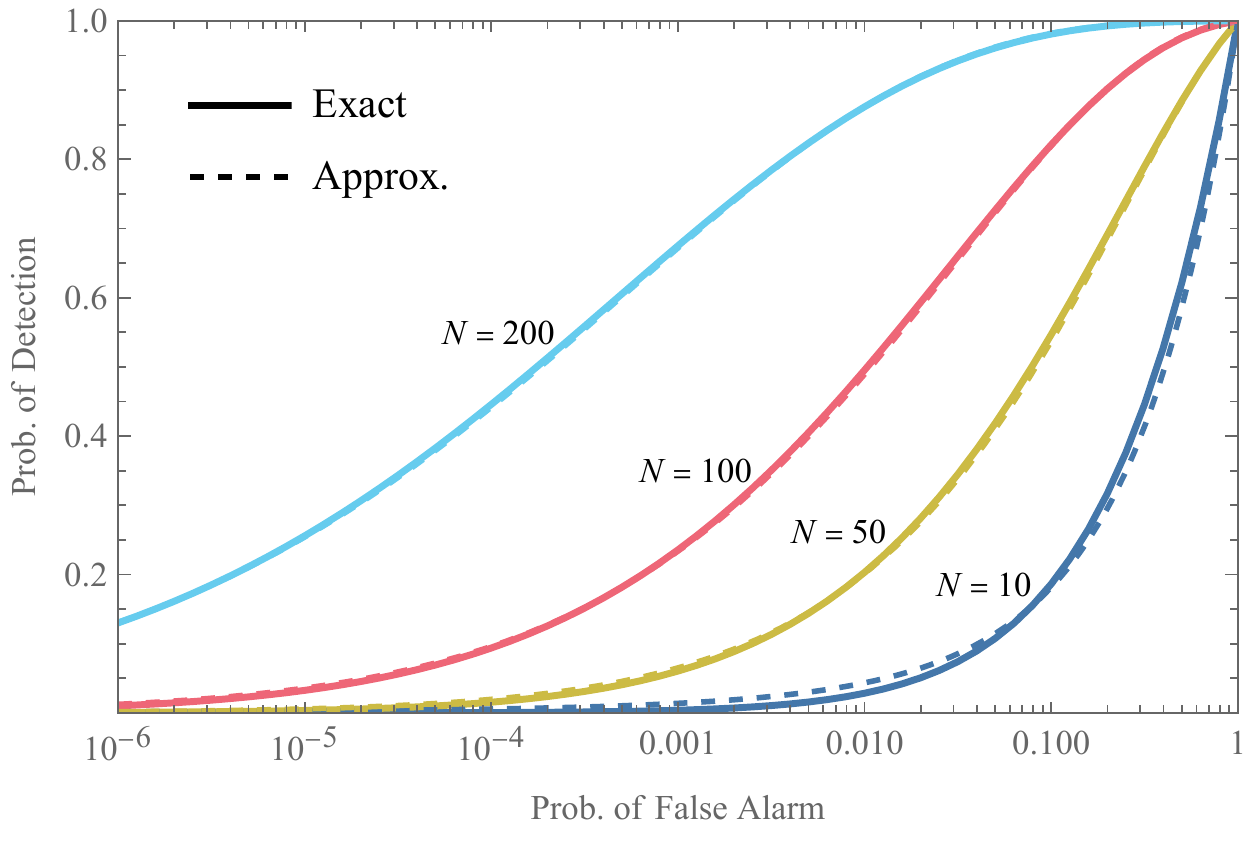}
		\label{subfig:ROC_rho_N}}
	\caption{ROC curves for $\hat{\rho}$, together with approximations calculated using \eqref{eq:ROC_rho}. In (a), $N = 10$ and $\rho \in \{0.2, 0.4, 0.6, 0.8\}$; in (b), $\rho = 0.2$ and $N \in \{10, 50, 100, 200\}$.}
	\label{fig:ROC_rho}
\end{figure*}

Fig.\ \ref{fig:ROC_rho} shows ROC curves for the $\hat{\rho}$ detector together with corresponding approximations obtained using \eqref{eq:ROC_rho}. In all cases, the approximation gives a fair idea of the behavior of the exact ROC curve. But even at $N = 50$---half the stated value of $N = 100$ for the validity of the approximation---the approximate curve is visually indistinguishable from the exact curve.

\subsection{Target Detection and MFN Estimation}

In \cite{dawood2001roc}, Dawood and Narayanan proposed and analyzed a design for a noise radar receiver which, in effect, calculates the detector
\begin{equation} \label{eq:det_DN}
	D_\mathrm{DN} = \frac{N}{4} \sqrt{\bar{R}_c^2 + \bar{R}_s^2}.
\end{equation}
Comparing this with \eqref{eq:est_rho}, the connection between $D_\mathrm{DN}$ and $\hat{\rho}$ is obvious. It bears a similar relation to $\hat{\rho}$ as covariance does to correlation, one being a normalized form of the other.

The main motivation for $D_\mathrm{DN}$ is that it arises naturally from performing matched filtering on the complex-valued signal $I_1[n] + j Q_1[n]$ using the reference signal $I_2[n] + j Q_2[n]$. However, it is interesting to note that $D_\mathrm{DN}$ can also be motivated using the MFN approach outlined in Sec.\ \ref{subsec:MFN_est}. One way is to calculate the norm of the difference between \eqref{eq:QTMS_cov} under the two hypotheses \eqref{eq:hypotheses}:
\begin{equation}
	\mleft\lVert \mat{\Sigma}(\sigma_1, \sigma_2, 0, \phi) - \mat{\Sigma}(\sigma_1, \sigma_2, \rho, \phi) \mright\rVert_F = 2 \rho \sigma_1 \sigma_2
\end{equation}
Substituting the MFN parameter estimates \eqref{eq:est_sigma1}--\eqref{eq:est_rho} yields $\sqrt{\bar{R}_c^2 + \bar{R}_s^2} = 4 D_\mathrm{DN}/N$. The factor of $4/N$, of course, does not affect the performance of the detector in any way.

Another way to see the connection between $D_\mathrm{DN}$ and MFN parameter estimation is inspired by the GLR test. Instead of calculating the difference between log-likelihoods, we calculate the difference between the squares of the minimized Frobenius norms:
\begin{align}
	&\min \, \mleft\lVert \mat{\Sigma}(\sigma_1, \sigma_2, 0, \phi) - \hat{\mat{S}} \mright\rVert_F^2 - \min \, \mleft\lVert \mat{\Sigma}(\sigma_1, \sigma_2, \rho, \phi) - \hat{\mat{S}} \mright\rVert_F^2 \nonumber \\
		&\qquad = \bar{R}_c^2 + \bar{R}_s^2.
\end{align}
The second line follows from \eqref{eq:g1_min} and \eqref{eq:g1_rho_0} in Appendix \ref{app:mfn}. This can be interpreted as the (squared) excess error that accrues from modeling the radar measurement data using the diagonal covariance matrix $\mat{\Sigma}(\sigma_1, \sigma_2, 0, \phi)$ as opposed to the more general form $\mat{\Sigma}(\sigma_1, \sigma_2, \rho, \phi)$. If the excess error is small, then the data is well-described by a diagonal covariance matrix and the target is probably absent, while the opposite is true if the excess error is large. And when considered as a detector, this excess error is equivalent to $D_\mathrm{DN}$.

For completeness, we quote the expressions for the PDF and CDF of $D_\mathrm{DN}$ that were derived by Dawood and Narayanan.

\begin{proposition}
	The PDF of $D_\mathrm{DN}$ for $x \geq 0$ is
	\vspace{-\jot}
	\begin{multline} \label{eq:det_DN_PDF}
		f_\text{DN}(x | \sigma_1, \sigma_2, \rho, N) = \\ 
			\frac{8 \tilde{x}^N}{\sigma_1\sigma_2 (1-\rho^2) \Gamma(N)} K_{N-1} \mleft( \frac{2\tilde{x}}{1-\rho^2} \mright) I_0 \mleft( \frac{2\rho\tilde{x}}{1-\rho^2} \mright)
	\end{multline}
	where $\tilde{x} \equiv 2x/(\sigma_1\sigma_2)$ and $K_{N-1}$ is the modified Bessel function of the second kind of order $N-1$. The CDF is
	\vspace{-\jot}
	\begin{multline} \label{eq:det_DN_CDF}
		F_\mathrm{DN}(x | \sigma_1, \sigma_2, \rho, N) = \\ 
			1 - \frac{2 \tilde{x}^N}{\Gamma(N)} \sum_{m=0}^{\infty} \rho^m K_{N+m} \mleft( \frac{2\tilde{x}}{1-\rho^2} \mright) I_m \mleft( \frac{2\rho\tilde{x}}{1-\rho^2} \mright).
	\end{multline}
\end{proposition}

\begin{proof}
	See Sec.\ V of \cite{dawood2001roc}.
\end{proof}

Like \eqref{eq:PDF_rho} and \eqref{eq:PDF_phi}, these expressions are rather cumbersome to work with. In the spirit of the approximations given in Propositions \ref{prop:approx_rice} and \ref{prop:approx_vonMises}, we now derive an approximate expression for the distribution of $D_\mathrm{DN}$. This time, however, we are able to supply a proof of the proposition.

\begin{proposition} \label{prop:approx_detDN}
	In the limit $N \to \infty$ and to first order in $\rho$, $D_\mathrm{DN}$ follows a Rice distribution with parameters  
	\begin{subequations}
		\begin{align}
			\alpha &= \frac{N}{2}\rho\sigma_1\sigma_2 \\
			\beta &= \sqrt{\frac{N}{8}}\sigma_1\sigma_2.
		\end{align}
	\end{subequations}
\end{proposition}

\begin{proof}
	According to the central limit theorem, the random vector $[\bar{R}_c, \bar{R}_s]^\T$ follows a bivariate normal distribution when $N \to \infty$:
	\begin{equation}
		\begin{bmatrix}
			\bar{R}_c \\
			\bar{R}_s
		\end{bmatrix} \sim 
		\mathcal{N} \mleft( 
		\begin{bmatrix}
			2\rho \sigma_1\sigma_2 \cos \phi \\
			2\rho \sigma_1\sigma_2 \sin \phi
		\end{bmatrix} \! ,
		\frac{\sigma_1^2 \sigma_2^2}{2N} [\mat{1}_2 + \rho^2 \mat{R}'(2\phi)]
		\mright).
	\end{equation}
	The mean vector is obtained by simply reading off and summing the appropriate entries in \eqref{eq:QTMS_cov}. The covariance matrix can be calculated by repeatedly applying \cite[Eq.\ (13)]{bohrnstedt1969cov}, which gives an expression for the expected value of fourth-order terms such as $\expval{I_1 I_2 Q_1 Q_2}$. It is evident that, to first order in $\rho$, the covariance matrix of $[\bar{R}_c, \bar{R}_s]^\T$ is proportional to the identity matrix.
	
	Recall that, for any $\theta$, the Rice distribution arises from the Euclidean norm of a bivariate normal random vector as follows:
	\begin{equation}
		\vec{X} \sim \mathcal{N} \mleft( 
		\begin{bmatrix}
			\alpha \cos \theta \\
			\alpha \sin \theta
		\end{bmatrix} \! , 
	\beta^2\mat{1}_2 \mright) \implies \| \vec{X} \| \sim \mathrm{Rice}(\alpha, \beta).
	\end{equation}
	Therefore, to first order in $\rho$, $\sqrt{\bar{R}_c^2 + \bar{R}_s^2}$ follows a Rice distribution with parameters $\alpha = 2\rho \sigma_1\sigma_2$ and $\beta = \sigma_1\sigma_2/\sqrt{2N}$ when $N \to \infty$. The proposition follows upon rescaling $\sqrt{\bar{R}_c^2 + \bar{R}_s^2}$ by a factor of $N/4$.
\end{proof}

\begin{remark}
	In \cite{dawood2001roc}, Dawood and Narayanan observe that when $\rho = 0$ and $N$ is large, $D_\mathrm{DN}$ is Rayleigh distributed with scale parameter $\sigma = \sigma_1\sigma_2\sqrt{N/8}$. The Rice distribution reduces to the Rayleigh distribution when $\alpha = 0$, so our result is in agreement with Dawood and Narayanan's observation.
\end{remark}

\begin{figure}[t]
	\centerline{\includegraphics[width=\columnwidth]{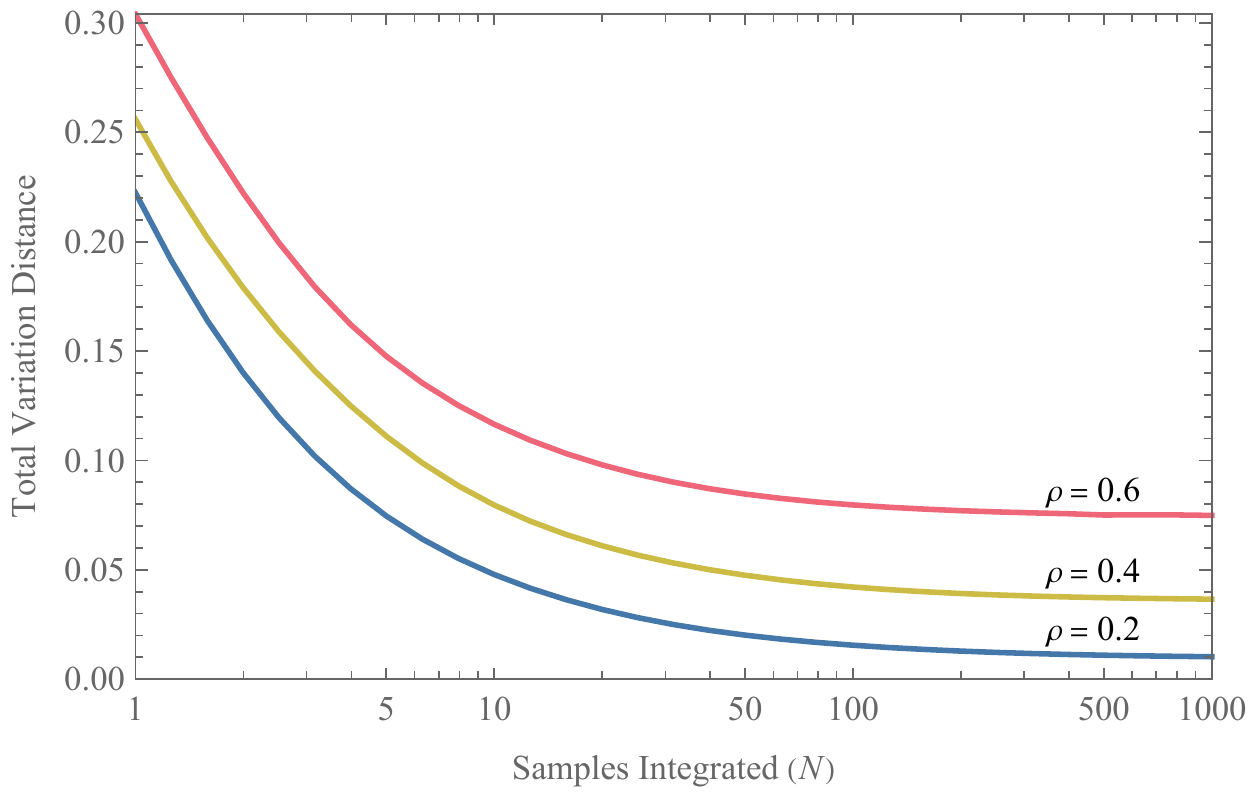}}
	\caption{Total variation distance between the exact probability density function of $D_\mathrm{DN}$ and the approximation described in Proposition \ref{prop:approx_detDN}, plotted as a function of $N$, for $\rho \in \{0.2, 0.4, 0.6\}$.}
	\label{fig:totVarDist_detDN}
\end{figure}

To quantify the goodness of the approximation in Proposition \ref{prop:approx_detDN}, we use the TVD as we did previously. Fig.\ \ref{fig:totVarDist_detDN} shows plots of $\mathit{TVD}_{D_\mathrm{DN}}$ as a function of $N$. We see that as $N$ becomes large, $\mathit{TVD}_{D_\mathrm{DN}}$ decreases to a steady-state value which increases with $\rho$. Hence, as expected, the Rician approximation becomes better when $N$ is large, but is only good when $\rho$ is small. Moreover, the smaller the value of $\rho$, the smaller the $N$ required for the approximation to be a good one.

\begin{figure}[t]
	\centering
	\subfloat[]{\includegraphics[width=\columnwidth]{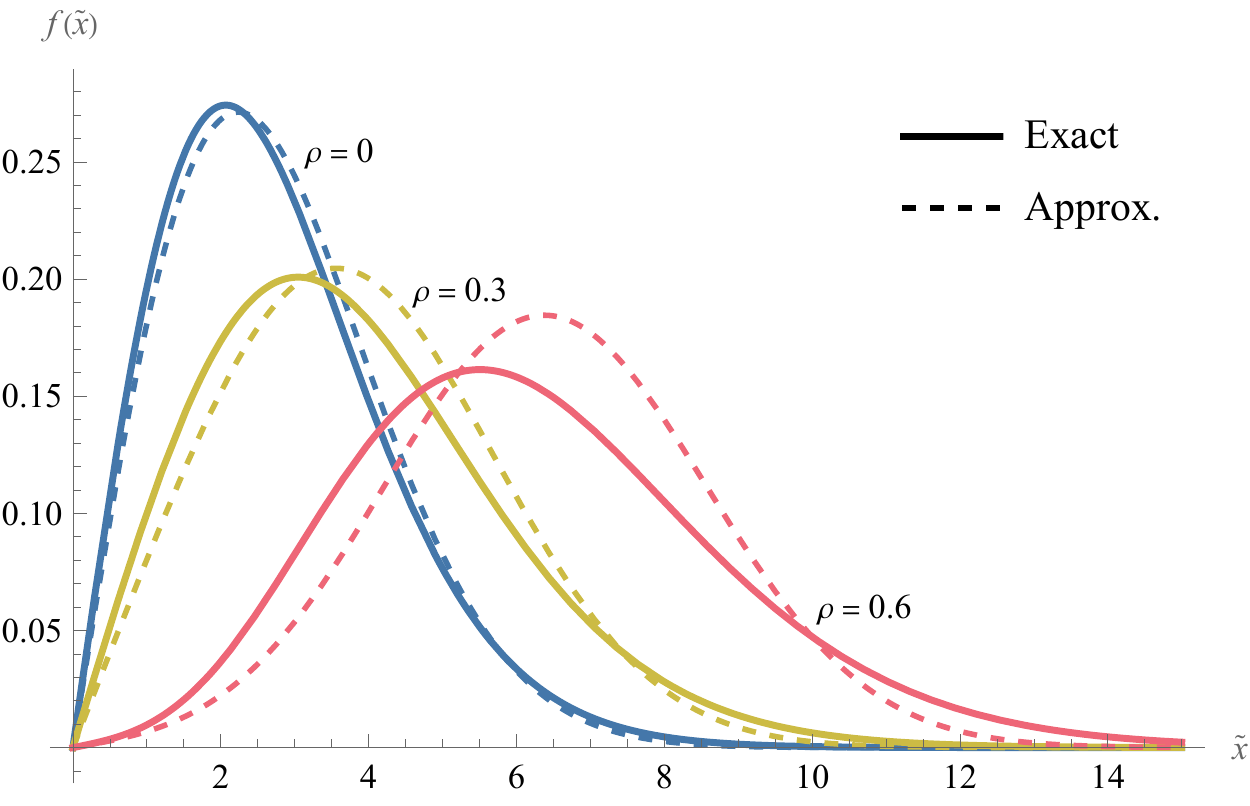}
		\label{subfig:PDF_detDN_rho}}
	\\
	\subfloat[]{\includegraphics[width=\columnwidth]{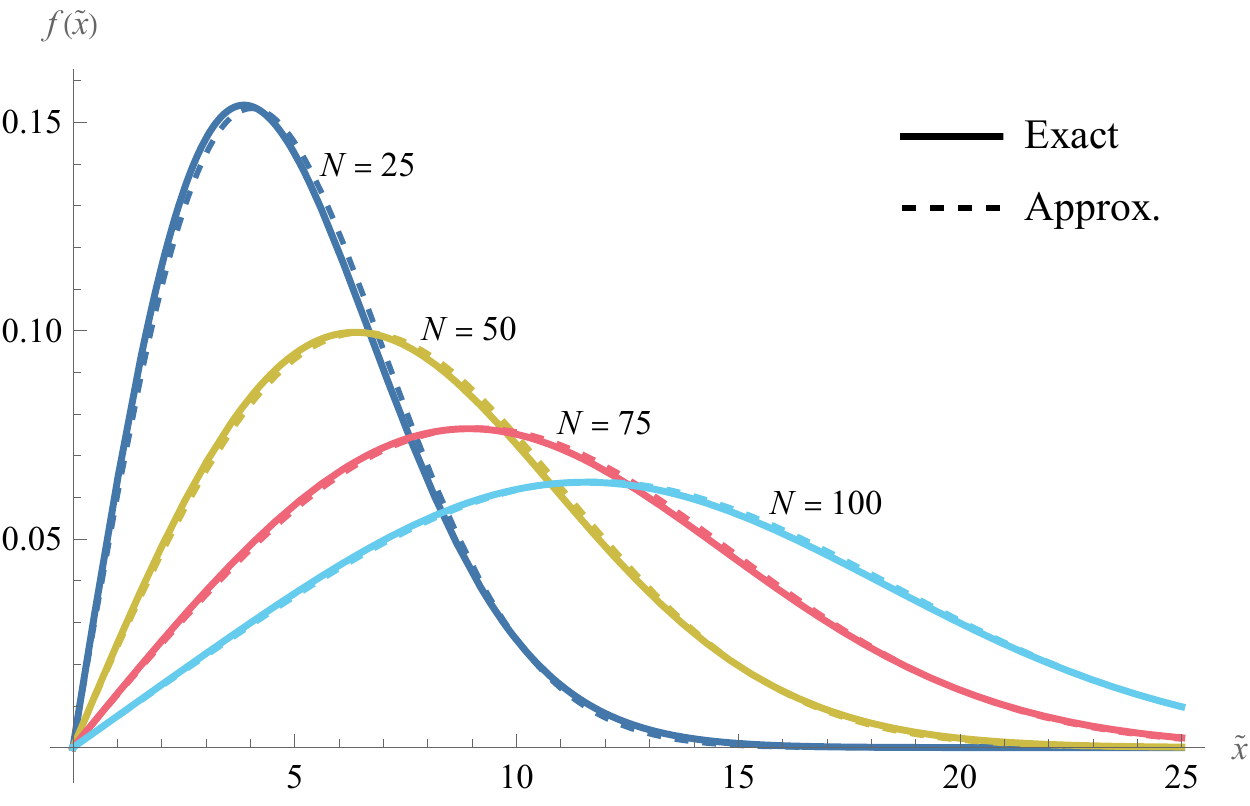}
		\label{subfig:PDF_detDN_N}}
	\caption{Probability density function of $D_\mathrm{DN}$ as a function of the normalized detector output $\tilde{x} \equiv 2x/(\sigma_1\sigma_2)$, together with the Rice distribution approximation described in Proposition \ref{prop:approx_detDN}. In (a), $N = 10$ and $\rho \in \{0, 0.3, 0.6\}$; in (b), $\rho = 0.1$ and $N \in \{25, 50, 75, 100\}$.}
	\label{fig:PDF_detDN}
\end{figure}

In Fig.\ \ref{fig:PDF_detDN}, we plot $f_\text{DN}(x | \sigma_1, \sigma_2, \rho, N)$ as a function of the normalized detector output $\tilde{x} \equiv 2x/(\sigma_1\sigma_2)$. By this normalization, we eliminate the need for separate plots in which $\sigma_1$ and $\sigma_2$ are varied, and we need only consider $\rho$ and $N$. The same figure also shows the corresponding Rice distribution approximations. Note that as $\rho$ increases, the approximation becomes worse and worse; conversely, as $N$ increases, the approximation becomes better and better.

Finally, we use the approximation in Proposition \ref{prop:approx_detDN} to give a closed-form approximation for the ROC curve of the $D_\mathrm{DN}$ detector.

\begin{proposition}
	In the limit $N \to \infty$ and to first order in $\rho$, the ROC curve for the detector $D_\mathrm{DN}$ is
	\begin{equation} \label{eq:ROC_detDN}
		\pd(\pfa | \rho, N) = Q_1 \mleft( \rho\sqrt{2N}, \sqrt{-2 \ln \smash{\pfa}} \mright).
	\end{equation}
\end{proposition}

\begin{proof}
	When the radar target is absent ($\rho = 0$), the Rice distribution reduces to the Rayleigh distribution, the CDF of which is well known. Using Proposition \ref{prop:approx_detDN}, it is easy to show that
	\begin{equation}
		\pfa(T) = \exp \mleft( -\frac{4 T^2}{N \sigma_1^2 \sigma_2^2} \mright).
	\end{equation}
	The remainder of the proof is the same as that of Proposition \ref{prop:approx_rice}, except that we use the parameters listed in Proposition \ref{prop:approx_detDN}.
\end{proof}

\begin{figure*}[t]
	\centering
	\subfloat[]{\includegraphics[width=\columnwidth]{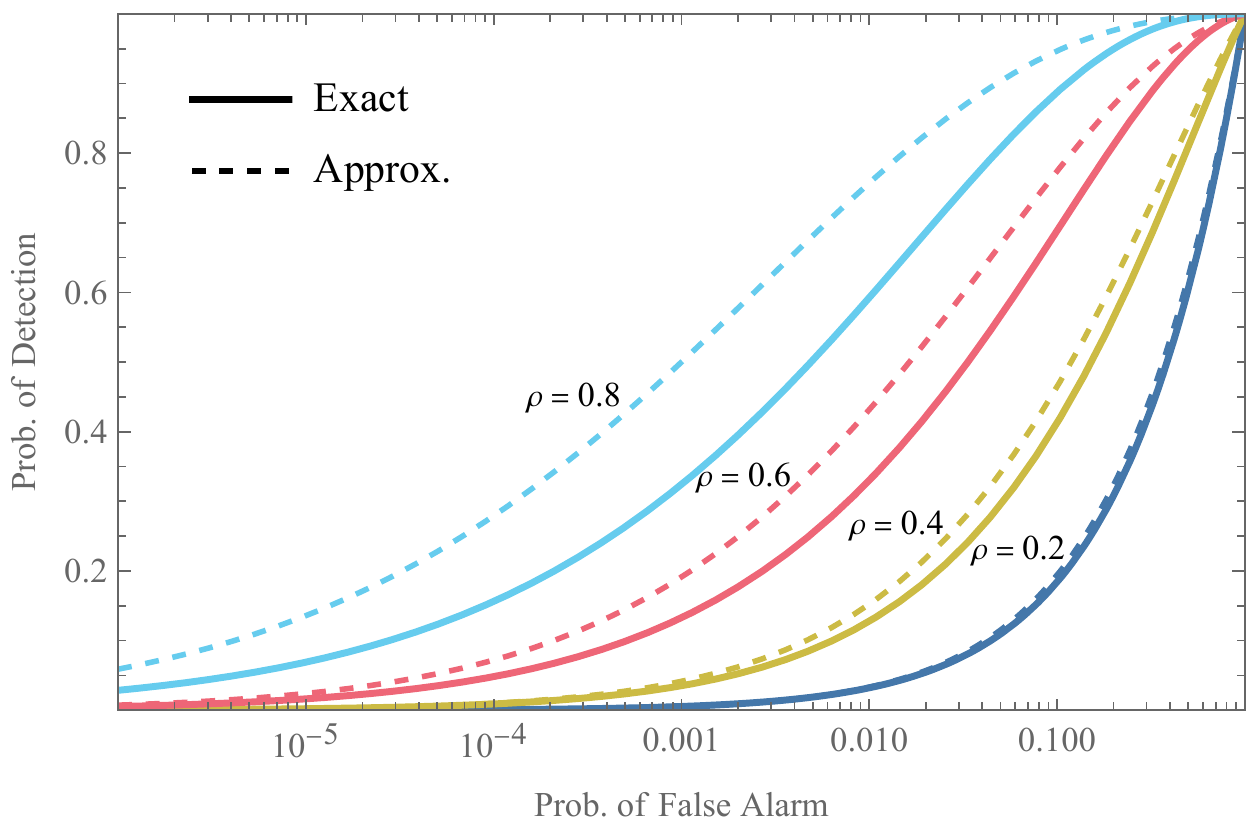}
		\label{subfig:ROC_detDN_rho}}
	\hfil
	\subfloat[]{\includegraphics[width=\columnwidth]{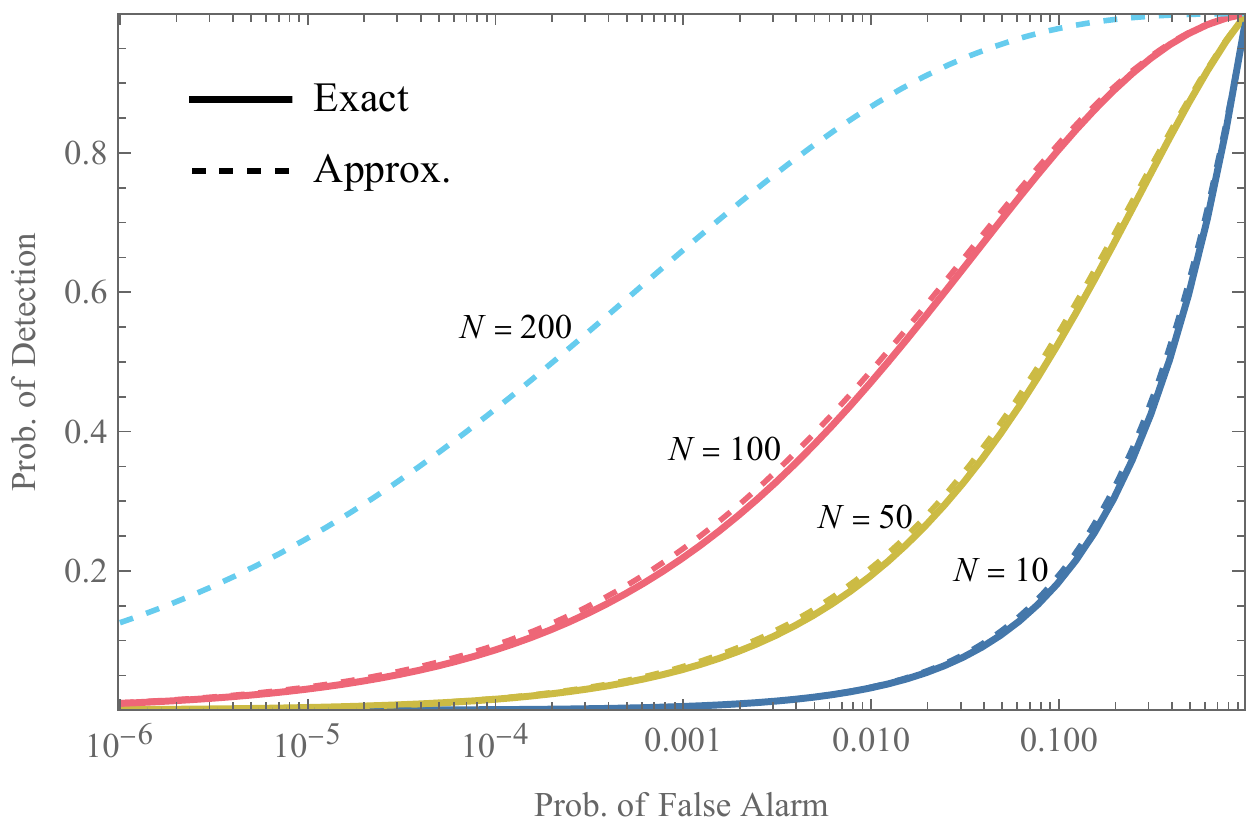}
		\label{subfig:ROC_detDN_N}}
	\caption{ROC curves for $D_\mathrm{DN}$, together with approximations calculated using \eqref{eq:ROC_detDN}. In (a), $N = 10$ and $\rho \in \{0.2, 0.4, 0.6, 0.8\}$; in (b), $\rho = 0.2$ and $N \in \{10, 50, 100, 200\}$. Due to numerical instabilities, the ROC curve for $\rho = 0.2$, $N = 200$ has been omitted, and only the approximation is shown.}
	\label{fig:ROC_detDN}
\end{figure*}

Fig.\ \ref{fig:ROC_detDN} shows ROC curve plots for the $D_\mathrm{DN}$ detector, together with approximations obtained from \eqref{eq:ROC_detDN}. We see that the approximation is good for small values of $\rho$, but \eqref{eq:ROC_detDN} overestimates the performance of the detector when $\rho$ is large. Incidentally, Fig.\ \ref{subfig:ROC_detDN_N} shows the value of the approximations derived in this paper: numerical instabilities prevented us from plotting the ROC curve for $\rho = 0.2$, $N = 200$, and we were only able to plot the approximate curve.

\subsection{Comparison of ROC Curves for $\hat{\rho}$ and $D_\mathrm{DN}$}

It should come as no surprise that the ROC curves for $\hat{\rho}$ and $D_\mathrm{DN}$ are the same when $N \to \infty$ and $\rho \ll 1$. To see this, consider the ROC curve for $\hat{\rho}$ in the form \eqref{eq:ROC_rho_old}; this is a good approximation to \eqref{eq:ROC_rho} when $N$ is large. When $\rho \ll 1$, the $\rho^2$ terms in \eqref{eq:ROC_rho_old} may be ignored; the result is exactly \eqref{eq:ROC_detDN}. Hence, under the stated conditions, the two detectors are essentially equivalent. 

We should note that the conditions $N \to \infty$ and $\rho \ll 1$ have more than a purely mathematical significance. In fact, the correlation coefficient $\rho$ is a decreasing function of range \cite{luong2022performance}; it also depends on factors such as the radar cross section of the target. Thus, the small-$\rho$ limit corresponds to the case where the target of the radar is small or far away. Under such conditions, the easiest way to compensate is by increasing the integration time---in other words, increasing $N$. (One could also compensate by increasing the transmit power; this would increase $\rho$ instead.) In summary, $\hat{\rho}$ and $D_\mathrm{DN}$ perform similarly when the target of the radar is small, far away, or otherwise difficult to detect. In this case, it may be preferable to use $D_\mathrm{DN}$, if only because \cite{dawood2001roc} includes an explicit block diagram showing how to build the detector using analog components such as mixers.

\begin{figure}[t]
	\centerline{\includegraphics[width=\columnwidth]{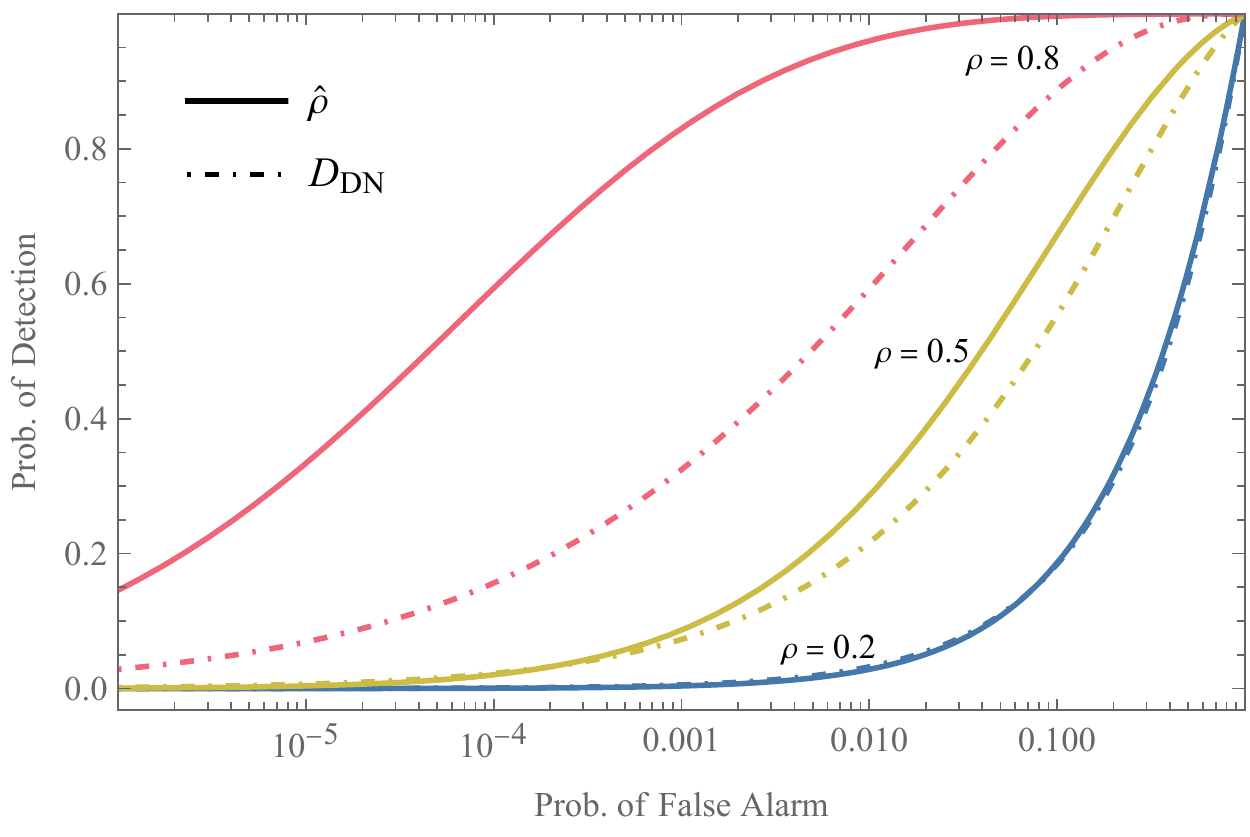}}
	\caption{Comparison of ROC curves for $\hat{\rho}$ and $D_\mathrm{DN}$ when $N = 10$ and $\rho \in \{0.2, 0.5, 0.8\}$.}
	\label{fig:ROC_comparison}
\end{figure}

At the opposite extreme, however, it turns out that the two detectors can behave quite differently. When $\rho$ is large and $N$ is small, it is possible for $\hat{\rho}$ to outperform $D_\mathrm{DN}$. In Fig.\ \ref{fig:ROC_comparison}, we plot (exact) ROC curves for the two detectors for $N = 10$. When $\rho = 0.2$, the two detectors remain indistinguishable, but as $\rho$ increases, $\hat{\rho}$ achieves a far higher $\pd$ for a given $\pfa$. Therefore, when it is desired to detect a nearby target quickly, it is advantageous to use $\hat{\rho}$.

\section{Conclusion}
\label{sec:conclusion}

This paper focused on deriving estimators for the four parameters that appear in the noise/QTMS radar covariance matrix \eqref{eq:QTMS_cov}, and elucidating certain statistical properties of these estimators. Our results may be summarized as follows: we derived estimators for the parameters, we characterized the probability distributions of the estimators, and we applied the results to the problem of target detection.

In Sec.\ \ref{sec:estimating}, we considered two methods for obtaining estimates of the parameters $\sigma_1$, $\sigma_2$, $\rho$, and $\phi$. One of them was based on minimizing the Frobenius norm between the sample covariance matrix (calculated directly from radar measurement data) and the structured matrix \eqref{eq:QTMS_cov}. The other was maximum likelihood estimation. Remarkably, both methods give the same estimates.

In Sec.\ \ref{sec:pdfs}, we gave expressions for the probability density functions for each of the four estimators. Another remarkable coincidence manifested here: for $\hat{\rho}$ and $\hat{\phi}$, we were able to reuse results from the theory of two-channel SAR, saving us the trouble of deriving the PDFs from scratch. Unfortunately, these PDFs were very complicated, involving the use of hypergeometric functions. However, we empirically found that these distributions could be approximated by much simpler distributions, namely the Rice distribution (for $\hat{\rho}$) and the von Mises distribution (for $\hat{\phi}$).

Finally, in Sec.\ \ref{sec:target_detection}, we applied our results to the noise radar target detection problem. We found that the generalized likelihood ratio test was equivalent to using $\hat{\phi}$ as a detector; we also showed connections between the minimum Frobenius norm method for parameter estimation and the detector $D_\mathrm{DN}$ previously studied by Dawood and Narayanan in \cite{dawood2001roc}. Using the approximations from the previous section, we found closed-form equations for the ROC curves of $\hat{\rho}$ and $D_\mathrm{DN}$.

In summary, this paper represents a broad overview of the basic statistical behavior of noise-type radars. We hope, in particular, that the various approximations will be found enlightening. The idea that $\hat{\rho}$ roughly follows a Rice distribution, for example, tells us more about $\hat{\rho}$ than the bare fact that it follows the exact PDF \eqref{eq:PDF_rho}. And from a more practical perspective, the estimators \eqref{eq:est_sigma1}--\eqref{eq:est_phi} are not computationally onerous, and should not be too difficult to incorporate into radar systems.

The results in this paper suggest several avenues for future research. For example, we assumed that all external noise was additive white Gaussian noise. It is necessary to test, using an experimental noise radar (or even a QTMS radar), how well that assumption holds up in practice. Another subject for future research is the properties of other parameters that could be estimated from radar data, such as bearing or range. Range, in particular, is related to phase, an estimator for which is given in \eqref{eq:est_phi}. The peculiar square-root/linear behavior of this estimator, as seen in Fig.\ \ref{fig:kappa_approx}, suggests that the statistical properties of any estimator of the radar range should be carefully studied. Finally, we were able to reuse several results from the theory of two-channel SAR in this paper. It would be fascinating if we could unearth a deeper mathematical connection between noise radars and SAR in future work.

\appendices

\section{Derivation of the Minimum Frobenius Norm Estimators}
\label{app:mfn}

For convenience, instead of performing the minimization \eqref{eq:minimization} directly, we will minimize the square of the norm. The squared Frobenius distance between the theoretical QTMS covariance matrix $\mat{\Sigma}(\sigma_1, \sigma_2, \rho, \phi)$ and the sample covariance matrix $\hat{\mat{S}}$ is
\begin{gather}
	\begin{aligned}
		&g_1(\sigma_1, \sigma_2, \rho, \phi) \equiv \mleft\| \mat{\Sigma}(\sigma_1, \sigma_2, \rho, \phi) - \hat{\mat{S}} \mright\|_F^2 \\
		&\qquad = 2 (\sigma_1^4 + 2\rho^2 \sigma_1^2 \sigma_2^2 + \sigma_2^4) - 2 (\bar{P}_1 \sigma_1^2 + \bar{P}_2 \sigma_2^2) \\
		&\qquad\phantom{=}\qquad - 4 \rho \sigma_1 \sigma_2 (\bar{R}_c \cos \phi + \bar{R}_s \sin \phi) + \| \hat{\mat{S}} \|_F^2.
	\end{aligned}
\end{gather}
The estimators are obtained by minimizing $g_1(\sigma_1, \sigma_2, \rho, \phi)$ subject to the conditions $0 \leq \sigma_1$, $0 \leq \sigma_2$, and $0 \leq \rho \leq 1$. Note that $\lVert \hat{\mat{S}} \rVert_F^2$ is a constant that does not depend on any of the four parameters.

The minimum of $g_1(\sigma_1, \sigma_2, \rho, \phi)$ must lie either at a stationary point or on the boundary of the parameter space over which we maximize. It turns out that the minimum does not occur on the boundary, but we will leave an analysis of the boundary for later and focus on the stationary points for now. The stationary points of $g_1(\sigma_1, \sigma_2, \rho, \phi)$ can be obtained by setting $\nabla g_1(\sigma_1, \sigma_2, \rho, \phi) = 0$ and solving for the parameters $\sigma_1$, $\sigma_2$, $\rho$, and $\phi$. The four elements of $\nabla g_1(\sigma_1, \sigma_2, \rho, \phi)$ are
\begin{subequations}
	\begin{align}
		\label{eq:g1_dsigma1}
		\frac{\partial g_1}{\partial \sigma_1} &= 4 \sigma_1 (2\sigma_1^2 + 2 \rho^2 \sigma_2^2 - \bar{P}_1) - 4 \rho \sigma_2 (\bar{R}_c \cos \phi + \bar{R}_s \sin \phi) \\
		\label{eq:g1_dsigma2}
		\frac{\partial g_1}{\partial \sigma_2} &= 4 \sigma_2 (2\sigma_2^2 + 2 \rho^2 \sigma_1^2 - \bar{P}_2) - 4 \rho \sigma_1 (\bar{R}_c \cos \phi + \bar{R}_s \sin \phi) \\
		\label{eq:g1_drho}
		\frac{\partial g_1}{\partial \rho} &=  \rho \sigma_1^2 \sigma_2^2 - 4 \sigma_1 \sigma_2 (\bar{R}_c \cos \phi + \bar{R}_s \sin \phi) \\
		\label{eq:g1_dphi}
		\frac{\partial g_1}{\partial \phi} &= 4 \rho \sigma_1 \sigma_2 (\bar{R}_c \sin \phi - \bar{R}_s \cos \phi).
	\end{align}
\end{subequations}
Solving $\partial g_1/\partial \phi = 0$ immediately yields the MFN estimator for $\phi$:
\begin{equation} \label{eq:g1_sol_phi}
	\hat{\phi} = \atantwo(\bar{R}_s, \bar{R}_c).
\end{equation}
Substituting this into \eqref{eq:g1_drho} and rearranging the equation $\partial g_1/\partial \rho = 0$ gives
\begin{equation} \label{eq:g1_rho_prelim}
	\rho = \frac{\sqrt{\bar{R}_c^2 + \bar{R}_s^2}}{2 \sigma_1 \sigma_2}
\end{equation}
Substituting \eqref{eq:g1_sol_phi} and \eqref{eq:g1_rho_prelim} into \eqref{eq:g1_dsigma1} yields
\begin{equation}
	0 = 8 \sigma_1^3 - 4 \bar{P}_1 \sigma_1,
\end{equation}
which yields the MFN estimator for $\sigma_1$:
\begin{equation}
	\hat{\sigma}_1 = \sqrt{ \frac{\bar{P}_1}{2} }.
\end{equation}
The MFN estimator for $\sigma_2$ can be obtained from \eqref{eq:g1_dsigma2} in exactly the same manner:
\begin{equation}
	\hat{\sigma}_2 = \sqrt{ \frac{\bar{P}_2}{2} }.
\end{equation}
Finally, substituting $\hat{\sigma}_1$ and $\hat{\sigma}_2$ into \eqref{eq:g1_rho_prelim} yields
\begin{equation}
	\hat{\rho} = \sqrt{ \frac{\bar{R}_c^2 + \bar{R}_s^2}{\bar{P}_1 \bar{P}_2} }.
\end{equation}

To complete the proof, we will now show that $g_1$ is not minimized on the boundaries of our optimization problem. First, note that
\begin{equation} \label{eq:g1_min}
	g_1(\hat{\sigma}_1, \hat{\sigma}_2, \hat{\rho}, \hat{\phi}) = \lVert \hat{\mat{S}} \rVert_F^2 - \frac{\bar{P}_1^2 + \bar{P}_2^2}{2} - \bar{R}_c^2 - \bar{R}_s^2.
\end{equation}
It is easy to show that in the case where $\sigma_1 = 0$,
\begin{equation}
	\min_{\sigma_2, \rho, \phi} g_1(0, \sigma_2, \rho, \phi) = \lVert \hat{\mat{S}} \rVert_F^2 - \frac{\bar{P}_2^2}{2}.
\end{equation}
This is manifestly greater than \eqref{eq:g1_min}, so the minimum does not occur when $\sigma_1 = 0$. A similar result occurs when $\sigma_2 = 0$. Likewise, when $\rho = 0$,
\begin{equation} \label{eq:g1_rho_0}
	\min_{\sigma_1, \sigma_2, \phi} g_1(\sigma_1, \sigma_2, 0, \phi) = \lVert \hat{\mat{S}} \rVert_F^2 - \frac{\bar{P}_1^2 + \bar{P}_2^2}{2}
\end{equation}
which again is greater than \eqref{eq:g1_min}, so the minimum does not occur when $\rho = 0$, either. The final case is $\rho = 1$, which in fact is a very complicated case requiring the use of a computer algebra system. Although we omit the relevant expressions here, we have verified that the minimum does not occur at $\rho = 1$. We may conclude, therefore, that the MFN estimators are indeed as given above.

\section{Derivation of the Maximum Likelihood Estimators}
\label{app:ml}

As shown in \cite{burg1982estimation}, maximizing the likelihood function is equivalent to maximizing the function
\vspace{-\jot}
\begin{multline}
		g_2(\sigma_1, \sigma_2, \rho, \phi) = \\
		-\ln |\mat{\Sigma}(\sigma_1, \sigma_2, \rho, \phi)| - \operatorname{tr} \mleft[ \mat{\Sigma}(\sigma_1, \sigma_2, \rho, \phi)^{-1} \hat{\mat{S}} \mright]
\end{multline}
As above, we impose the conditions $0 \leq \sigma_1$, $0 \leq \sigma_2$, and $0 \leq \rho \leq 1$. By a straightforward but tedious calculation, we find that
\vspace{-\jot}
\begin{multline}
	g_2(\sigma_1, \sigma_2, \rho, \phi) = - 2 \ln \mleft[ \sigma_1^2 \sigma_2^2 (1 - \rho^2) \mright] \\
		- \frac{1}{1 - \rho^2} \mleft( \frac{\bar{P}_1}{\sigma_1^2} + \frac{\bar{P}_2}{\sigma_2^2} - \frac{2 \rho (\bar{R}_c \cos \phi + \bar{R}_s \sin \phi)}{\sigma_1 \sigma_2} \mright)
\end{multline}
The maximum of $g_2(\sigma_1, \sigma_2, \rho, \phi)$ must lie either at a stationary point or on the boundary of the parameter space over which we maximize. Some parts of the boundary are easily taken care of: when $\sigma_1 = 0$, $\sigma_2 = 0$, or $\rho = 1$, $g_2(\sigma_1, \sigma_2, \rho, \phi)$ is undefined, so no maximum can occur at those points. This leaves only $\rho = 0$. For now, we will assume $\rho \neq 0$ and return to this case later.

The stationary points of $g_2(\sigma_1, \sigma_2, \rho, \phi)$ can be obtained by setting $\nabla g_2(\sigma_1, \sigma_2, \rho, \phi) = 0$ and solving for the four parameters. The elements of $\nabla g_2(\sigma_1, \sigma_2, \rho, \phi)$ are
\begin{subequations}
	\begin{align}
		\label{eq:g2_dsigma1}
		\frac{\partial g_2}{\partial \sigma_1} &= -\frac{4}{\sigma_1} + \frac{2}{1 - \rho^2} \mleft( \frac{\bar{P}_1}{\sigma_1^3} - \frac{\rho (\bar{R}_c \cos \phi + \bar{R}_s \sin \phi)}{\sigma_1^2 \sigma_2} \mright) \\
		\label{eq:g2_dsigma2}
		\frac{\partial g_2}{\partial \sigma_2} &= -\frac{4}{\sigma_2} + \frac{2}{1 - \rho^2} \mleft( \frac{\bar{P}_2}{\sigma_2^3} - \frac{\rho (\bar{R}_c \cos \phi + \bar{R}_s \sin \phi)}{\sigma_1 \sigma_2^2} \mright) \\
		\label{eq:g2_drho}
		\frac{\partial g_2}{\partial \rho} &= \frac{4\rho}{1 - \rho^2} + \frac{2(\bar{R}_c \cos \phi + \bar{R}_s \sin \phi)}{\sigma_1 \sigma_2 (1 - \rho^2)} \nonumber \\
			 &\phantom{=}\ - \frac{2 \rho}{(1 - \rho^2)^2} \mleft( \frac{\bar{P}_1}{\sigma_1^2} + \frac{\bar{P}_2}{\sigma_2^2} - \frac{2 \rho (\bar{R}_c \cos \phi + \bar{R}_s \sin \phi)}{\sigma_1 \sigma_2} \mright) \\
		\label{eq:g2_dphi}
		\frac{\partial g_2}{\partial \phi} &=	\frac{2 \rho (\bar{R}_s \cos \phi - \bar{R}_c \sin \phi)}{\sigma_1 \sigma_2 (1 - \rho^2)}.
	\end{align}
\end{subequations}
To begin, note that $\partial g_2/\partial \phi = 0$ can be solved immediately to yield the ML estimator for $\phi$:
\begin{equation} \label{eq:g2_sol_phi}
	\hat{\phi} = \atantwo(\bar{R}_s, \bar{R}_c).
\end{equation}
Next, we combine \eqref{eq:g2_dsigma1} and \eqref{eq:g2_dsigma2} as follows:
\begin{align} \label{eq:P_1_P_2}
	0 &= \sigma_1 \frac{\partial g_2}{\partial \sigma_1} - \sigma_2 \frac{\partial g_2}{\partial \sigma_2} \nonumber \\
		&= \frac{2}{1 - \rho^2} \mleft( \frac{\bar{P}_1}{\sigma_1^2} - \frac{\bar{P}_2}{\sigma_2^2} \mright)
\end{align}
It follows that $\sigma_2 = \sigma_1 \sqrt{\bar{P}_2/\bar{P}_1}$. Substituting this and \eqref{eq:g2_sol_phi} into \eqref{eq:g2_drho}, we find that, up to an unimportant prefactor,
\begin{equation}
	0 = \sqrt{ \frac{\bar{R}_c^2 + \bar{R}_s^2}{\bar{P}_1 \bar{P}_2} } (1 + \rho^2) + \frac{2 \sigma_1^2 \rho (1 - \rho^2)}{\bar{P}_1} - 2 \rho.
\end{equation}
Rearranging, we obtain
\begin{equation}
	\sigma_1^2 = \bar{P}_1 \mleft[ \frac{1}{1 - \rho^2} - \sqrt{ \frac{\bar{R}_c^2 + \bar{R}_s^2}{\bar{P}_1 \bar{P}_2} } \frac{1 + \rho^2}{2 \rho (1 - \rho^2)} \mright].
\end{equation}
Substituting this, \eqref{eq:g2_sol_phi}, and \eqref{eq:P_1_P_2} into \eqref{eq:g2_dsigma1} yields, after much simplification,
\begin{equation}
	0 = \bar{R}_c^2 + \bar{R}_s^2 - \rho \sqrt{ \bar{P}_1 \bar{P}_2 (\bar{R}_c^2 + \bar{R}_s^2) }.
\end{equation}
From this, we obtain the ML estimator for $\rho$:
\begin{equation} \label{eq:sol_rho}
	\hat{\rho} = \sqrt{ \frac{\bar{R}_c^2 + \bar{R}_s^2}{\bar{P}_1 \bar{P}_2} }.
\end{equation}
Once we substitute \eqref{eq:g2_sol_phi}, \eqref{eq:P_1_P_2}, and \eqref{eq:sol_rho} into \eqref{eq:g2_dsigma2}, we find that
\begin{equation}
	0 = \bar{P}_2 - 2\sigma_2^2.
\end{equation}
The ML estimators for $\sigma_1$ and $\sigma_2$ follow immediately:
\begin{align}
	\hat{\sigma}_1 &= \sqrt{ \frac{\bar{P}_1}{2} } \\
	\hat{\sigma}_2 &= \sqrt{ \frac{\bar{P}_2}{2} }.
\end{align}

We now return to the possibility that the maximum of $g_2(\sigma_1, \sigma_2, \rho, \phi)$ may occur at the boundary where $\rho = 0$. It turns out that in this case, the estimators $\hat{\sigma}_1$ and $\hat{\sigma}_2$ remain the same; this is easily verified by substituting $\rho = 0$ into \eqref{eq:g2_dsigma1} and \eqref{eq:g2_dsigma2} and solving. As for $\hat{\phi}$, it loses all meaning because $\phi$ does not enter into the likelihood function when $\rho = 0$. But which estimator, $\hat{\rho} = 0$ or \eqref{eq:sol_rho}, actually maximizes $g_2(\sigma_1, \sigma_2, \rho, \phi)$? This is exactly the question that the likelihood ratio detector \eqref{eq:D_GLR} is designed to answer. Therefore, the appropriate estimator for $\rho$ depends on whether the target is predicted to be present or absent: if present, use \eqref{eq:sol_rho}; if absent, $\hat{\rho} = 0$.

\section*{Acknowledgment}

This work was supported by the Natural Science and Engineering Research Council of Canada (NSERC). D.\ Luong also acknowledges the support of a Vanier Canada Graduate Scholarship.

\bibliographystyle{IEEEtran}
\bibliography{qradar_refs,own_refs}

\end{document}